\documentclass{article}

\usepackage{PRIMEarxiv}

\usepackage[utf8]{inputenc} 
\usepackage[T1]{fontenc}    
\usepackage{hyperref}       
\usepackage{url}            
\usepackage{booktabs}       
\usepackage{amsfonts}       
\usepackage{nicefrac}       
\usepackage{microtype}      
\usepackage{lipsum}
\usepackage{fancyhdr}       
\usepackage{graphicx}       
\usepackage{tikz-cd}
\usepackage{amsmath,amssymb,amsthm,mathrsfs}
\theoremstyle{plain}
\newtheorem{theorem}{Theorem}
\newtheorem{lemma}[theorem]{Lemma}
\newtheorem{corollary}[theorem]{Corollary}
\newtheorem{example}[theorem]{Example}
\theoremstyle{definition}
\newtheorem{definition}[theorem]{Definition}

\graphicspath{{media/}}     

\title{Generator polynomial matrices of the Galois hulls of multi-twisted codes
}

\author{
  Ramy F. Taki Eldin \\
  Faculty of Engineering, Ain Shams University, Cairo, Egypt\\
  \texttt{ramy.farouk@eng.asu.edu.eg} \\
   \AND
   Patrick Sol\'e \\
   I2M, (Aix Marseille Univ., CNRS, Centrale Marseille), Marseilles, France \\
   \texttt{sole@enst.fr} \\
}

\begin{document}
\maketitle
\sloppy

\begin{abstract}
In this study, we consider the Euclidean and Galois hulls of multi-twisted (MT) codes over a finite field $\mathbb{F}_{p^e}$ of characteristic $p$. Let $\mathbf{G}$ be a generator polynomial matrix (GPM) of a MT code $\mathcal{C}$. For any $0\le \kappa<e$, the $\kappa$-Galois hull of $\mathcal{C}$, denoted by $h_\kappa\left(\mathcal{C}\right)$, is the intersection of $\mathcal{C}$ with its $\kappa$-Galois dual. The main result in this paper is that a GPM for $h_\kappa\left(\mathcal{C}\right)$ has been obtained from $\mathbf{G}$. We start by associating a linear code $\mathcal{Q}_\mathbf{G}$ with $\mathbf{G}$. We show that $\mathcal{Q}_\mathbf{G}$ is quasi-cyclic. In addition, we prove that the dimension of $h_\kappa\left(\mathcal{C}\right)$ is the difference between the dimension of $\mathcal{C}$ and that of $\mathcal{Q}_\mathbf{G}$. Thus the determinantal divisors are used to derive a formula for the dimension of $h_\kappa\left(\mathcal{C}\right)$. Finally, we deduce a GPM formula for $h_\kappa\left(\mathcal{C}\right)$. In particular, we handle the cases of $\kappa$-Galois self-orthogonal and linear complementary dual MT codes; we establish equivalent conditions that characterize these cases. Equivalent results can be deduced immediately for the classes of cyclic, constacyclic, quasi-cyclic, generalized quasi-cyclic, and quasi-twisted codes, because they are all special cases of MT codes. Some numerical examples, containing optimal and maximum distance separable codes, are used to illustrate the theoretical results.
\end{abstract}


\section{Introduction}
Multi-twisted (MT) codes constitute an important class of linear codes, which first appeared in \cite{Aydin2017}. They generalize the class of cyclic codes in a very broad sense. A MT code of index $\ell$ and shift constants $\Lambda=\left(\lambda_1,\ldots,\lambda_\ell\right)$ is a linear code that is kept invariant under some linear transformation, denoted by $\mathcal{T}_\Lambda$. This linear transformation has the effect of constacyclic shifting each block of length $m_j$ by the shift constant $\lambda_j$ for $1\le j\le\ell$. MT codes contain several subclasses of codes that have been studied extensively in the literature. For instance, quasi-twisted (QT) codes, introduced in \cite{Jia2012}, are MT codes with equal block lengths and equal shift constants. While quasi-cyclic (QC) codes are MT codes with equal block lengths and $\lambda_j=1$ for all $1\le j\le\ell$. Generalizing QC codes by just one step results in the class of generalized quasi-cyclic (GQC) codes, which was introduced in \cite{Esmaeili2009,Kulhan2005}. GQC codes are MT codes in which $\lambda_j=1$ for all $1\le j\le\ell$. However, when the index is one, MT codes coincide with constacyclic codes. Various algebraic structures for MT codes have been introduced in the literature (see \cite{Aydin2017,Chauhan2021,Eldin1}). In \cite{Eldin1}, the invariance under $\mathcal{T}_\Lambda$ endows a MT code with the structure of a free module over the principal ideal domain (PID) $\mathbb{F}_q[x]$ (with $x$ acting by $\mathcal{T}_\Lambda$). A generator polynomial matrix (GPM), denoted $\mathbf{G}$, of a MT code $\mathcal{C}$ is a matrix over $\mathbb{F}_q[x]$ whose rows form a basis of $\mathcal{C}$ as an $\mathbb{F}_q[x]$-module. Whereas the reduced GPM of $\mathcal{C}$ is the unique Hermite normal form \cite{Eldin1} of any GPM of $\mathcal{C}$.

The study of duals of MT codes and their subclasses (e.g., GQC and constacyclic codes) has received much attention in the literature. The duals of cyclic, constacyclic, QC, QT, GQC, and MT codes have been studied in \cite{Fan2016,Yang2013,Saleh2016,Gneri2017,Sharma2018}, where these duals have been found to be cyclic, constacyclic, QC, QT, GQC, and MT, respectively. The Euclidean dual $\mathcal{C}^\perp$ of a linear code $\mathcal{C}$ is defined using the Euclidean inner product, while the $\kappa$-Galois dual $\mathcal{C}^{\perp_\kappa}$ (for $\kappa\ge0$) is defined similarly with the replacement of the Euclidean inner product with the $\kappa$-Galois inner product. The concept of the Galois inner product was first introduced in \cite{Fan2016} as a generalization of the Euclidean inner product. In \cite{Eldin2}, a GPM formula for the Euclidean dual of a MT code is illustrated. However, the Galois dual of a MT code is discussed in \cite{Chauhan2021,Eldin2}. For a linear code $\mathcal{C}$, the Euclidean hull of $\mathcal{C}$, denoted $h\left( \mathcal{C}\right)$, is the intersection of $\mathcal{C}$ with $\mathcal{C}^\perp$. While the $\kappa$-Galois hull of $\mathcal{C}$, denoted $h_\kappa\left(\mathcal{C}\right)$, is the intersection of $\mathcal{C}$ with $\mathcal{C}^{\perp_\kappa}$. Some optimal codes were obtained in \cite{Sendrier} by studying linear codes with trivial Euclidean hulls, which are known as linear complementary dual (LCD) codes. Later, Liu et al.\ \cite{Hongwei} studied the Galois hulls of linear codes. They gave a formula to determine the dimension of $h_\kappa\left(\mathcal{C}\right)$ by means of any generator matrix of a linear code $\mathcal{C}$.

This paper mainly focuses on the investigation of the Galois hulls of MT codes. Based on our conviction that GPMs are a compact way of algebraically representing MT codes, we aim to determine the Galois hulls of these codes by means of their GPMs. For any $\kappa\ge 0$, we begin by expressing the Galois inner product of two vectors in terms of their polynomial vector representation. For a MT code $\mathcal{C}$ with a GPM $\mathbf{G}$, we introduce a formula for a GPM for the $\kappa$-Galois dual of $\mathcal{C}$, which we denote by $\mathbf{H}_\kappa$. Then we use $\mathbf{G}$ to construct a polynomial matrix, denoted $\mathfrak{B}_\mathbf{G}$, that associates to $\mathbf{G}$ a linear code which will be denoted by $\mathcal{Q}_\mathbf{G}$. We prove that $\mathcal{Q}_\mathbf{G}$ is QC and has a dimension equal to the difference between the dimension of $\mathcal{C}$ and the dimension of $h_\kappa\left(\mathcal{C}\right)$. Hence we propose a formula for the dimension of $h_\kappa\left(\mathcal{C}\right)$ through the determinantal divisors of $\mathfrak{B}_\mathbf{G}$. It is also shown how the polynomial matrix $\mathfrak{B}_\mathbf{G}$ can be used to construct a GPM for $\mathcal{Q}_\mathbf{G}$, denoted $\mathfrak{G}_\mathbf{G}$. These results contributed to our analysis of the Galois self-orthogonal and LCD codes. Specifically, $\mathcal{C}$ is $\kappa$-Galois self-orthogonal if $h_\kappa\left(\mathcal{C}\right)=\mathcal{C}$, whereas $\mathcal{C}$ is $\kappa$-Galois LCD if $h_\kappa\left(\mathcal{C}\right)=\{\mathbf{0}\}$. We prove that $\mathcal{C}$ is $\kappa$-Galois self-orthogonal if and only if $\mathfrak{B}_\mathbf{G}=\mathbf{0}$. Likewise, we prove necessary and sufficient conditions for a MT code to be $\kappa$-Galois LCD.

Although the above result is a reformulation of \cite[Theorem 2]{Hongwei} for MT codes, the main and most important result in this paper can be found in Theorem \ref{main_theorem}, where we go beyond determining the dimension of the Galois hull of a MT code to provide its generator matrix as well. In the literature, such generator matrix has not been shown in particular for MT codes, nor in general for linear codes. 
Briefly, we show that if $\mathfrak{H}_\mathbf{G}$ is a GPM for the $\kappa$-Galois dual of $\mathcal{Q}_\mathbf{G}$, then $\mathfrak{H}_\mathbf{G} \mathbf{G}$ is a GPM for $h_\kappa\left(\mathcal{C}\right)$. Several numerical examples are devoted to illustrating the theoretical results developed in the paper. Some of these examples present MT codes that achieve the best known parameters of linear codes as recorded in \cite{Grassl}. In fact, some codes are near-optimal, others are optimal, and others are maximum distance separable (MDS).

The material is arranged as follows. Section \ref{Sec2} presents some preliminaries for constructing the algebraic structures of MT codes.
Section \ref{Sec3} is devoted to a discussion of the properties of the Euclidean and Galois duals of MT codes. The theoretical results and their numerical examples are presented in Section \ref{main_results_section}. Section \ref{conclusion} concludes the article.

\section{Representation of MT codes}
\label{Sec2}
Details of most of the material presented in this section can be found in \cite{Eldin1,Eldin2}. We will adopt the following notation throughout the paper: $\mathbb{F}_q$ denotes a finite field of order $q$, $\mathbb{F}_q^n$ is the vector space of all $n$-tuples of elements of $\mathbb{F}_q$, $\mathbb{F}_q[x]$ is the polynomial ring over $\mathbb{F}_q$, and $\langle g(x)\rangle$ is the ideal of $\mathbb{F}_q[x]$ generated by $g(x)$. Let $\mathcal{C}$ be a linear code over $\mathbb{F}_q$ of length $n$, or equivalently, a subspace of $\mathbb{F}_q^n$. We call $\mathcal{C}$ cyclic if it is invariant under the cyclic shift which is given by the linear transformation 
\begin{equation*}
\mathcal{T}: \left( a_0, a_1, \ldots, a_{n-2}, a_{n-1}\right) \mapsto \left( a_{n-1}, a_0, a_1, \ldots, a_{n-2}\right)
\end{equation*}
on $\mathbb{F}_q^n$. The polynomial representation of $\left( a_0, a_1, \ldots, a_{n-2}, a_{n-1}\right) \in \mathbb{F}_q^n$ is the element of the quotient ring $\mathbb{F}_q[x]/\langle x^n-1\rangle$ whose representative is $\sum_{i=0}^{n-1} a_i x^i$. This map induces an $\mathbb{F}_q$-vector space isomorphism between $\mathbb{F}_q^n$ and $\mathbb{F}_q[x]/\langle x^n-1\rangle$. Because of its invariance under $\mathcal{T}$, a cyclic code has the structure of an ideal in $\mathbb{F}_q[x]/\langle x^n-1\rangle$. However, it is known that there is a one-to-one correspondence between ideals of $\mathbb{F}_q[x]/\langle x^n-1\rangle$ and ideals of $\mathbb{F}_q[x]$ that contain $\langle x^n-1\rangle$. A cyclic code over $\mathbb{F}_q$ of length $n$ can then be viewed as an ideal in $\mathbb{F}_q[x]$ containing $\langle x^n-1\rangle$. By a cyclic code $\mathcal{C}$ we may sometimes mean an invariant subspace of $\mathbb{F}_q^n$ under $\mathcal{T}$, and other times we may mean an ideal in $\mathbb{F}_q[x]$ containing $\langle x^n-1\rangle$; this will not cause any confusion as it will be clear from the context what is meant. In the latter representation there exists a polynomial $g(x)\in\mathbb{F}_q[x]$ such that $\mathcal{C}=\langle g(x)\rangle$ because $\mathbb{F}_q[x]$ is a PID. In addition, $g(x)$ divides $x^n-1$ because $\mathcal{C}\supseteq \langle x^n-1\rangle$. That is, $a(x)g(x)=x^n-1$ for some $a(x)\in\mathbb{F}_q[x]$. One can show that $\mathcal{C}$ is a vector space over $\mathbb{F}_q$ of dimension 
$$\mathrm{dim}\left( \mathcal{C}\right)=\deg\left(a(x)\right)=n-\deg\left(g(x)\right).$$

By the previous paragraph, it will be evident how QC codes generalize cyclic codes. A linear code $\mathcal{C}$ over $\mathbb{F}_q$ of length $m\ell$ is QC if it is invariant under the cyclic shift with $\ell$ coordinates. In other words, $\mathcal{C}$ is invariant under the linear transformation
\begin{equation}
\label{cyclic_shift_QC}
\begin{split}
&\mathcal{F}: \left(a_{0,1},a_{0,2},\ldots,a_{0,\ell},a_{1,1},a_{1,2},\ldots,a_{1,\ell},\ldots,a_{m-1,1},a_{m-1,2},\ldots,a_{m-1,\ell}\right)\\ 
& \mapsto \left(a_{m-1,1},a_{m-1,2},\ldots,a_{m-1,\ell}, a_{0,1},a_{0,2},\ldots,a_{0,\ell},\ldots,a_{m-2,1},a_{m-2,2},\ldots,a_{m-2,\ell}\right) 
\end{split}
\end{equation}
on $\mathbb{F}_q^{m\ell}$. We call $\ell$ and $m$ the index and co-index of $\mathcal{C}$ respectively. Observe that a cyclic code is QC with $\ell=1$. We have an $\mathbb{F}_q$-vector space isomorphism between $\mathbb{F}_q^{m\ell}$ and $\oplus_{j=1}^{\ell} \mathbb{F}_q[x]/\langle x^m-1\rangle$ given by
\begin{equation}
\label{VS_Isomorphism}
\begin{split}
\left(a_{0,1},a_{0,2},\ldots,a_{0,\ell},a_{1,1},a_{1,2},\ldots,a_{1,\ell},\ldots,a_{m-1,1},a_{m-1,2},\ldots,a_{m-1,\ell}\right) \\  \mapsto \left(a_1\left(x\right)+\langle x^m-1\rangle,a_2\left(x\right)+\langle x^m-1\rangle,\ldots,a_\ell\left(x\right)+\langle x^m-1\rangle\right)
\end{split}\end{equation}
where $a_j\left(x\right)=\sum_{i=0}^{m-1} a_{i,j}x^i$ for $1\le j\le \ell$. The polynomial representation of $\mathbf{a}\in\mathbb{F}_q^{m\ell}$ is the polynomial vector
\begin{equation*}
\varphi\left(\mathbf{a}\right)=\left( \sum_{i=0}^{m-1} a_{i,1}x^i , \sum_{i=0}^{m-1} a_{i,2}x^i, \ldots , \sum_{i=0}^{m-1} a_{i,\ell}x^i\right) \in \oplus_{j=1}^{\ell}\mathbb{F}_q[x].
\end{equation*}
One can regard $\mathbb{F}_q^{m\ell}$ as an $\mathbb{F}_q[x]$-module via the scalar multiplication $x\cdot \mathbf{a}=\mathcal{F}\left(\mathbf{a}\right)$. Then the vector space isomorphism given by \eqref{VS_Isomorphism} is actually an $\mathbb{F}_q[x]$-module isomorphism. A QC code $\mathcal{C}$ is an $\mathcal{F}$-invariant subspace of $\mathbb{F}_q^{m\ell}$, and hence $\mathcal{C}$ may be regarded as an $\mathbb{F}_q[x]$-submodule of $\mathbb{F}_q^{m\ell}$, or, equivalently, as an $\mathbb{F}_q[x]$-submodule of $\oplus_{j=1}^{\ell} \mathbb{F}_q[x]/\langle x^m-1\rangle$. It is known that there is a one-to-one correspondence between submodules of $\oplus_{j=1}^{\ell} \mathbb{F}_q[x]/\langle x^m-1\rangle$ and submodules of $\oplus_{j=1}^{\ell}\mathbb{F}_q[x]$ that contain $\oplus_{j=1}^{\ell}\langle x^m-1\rangle$. Henceforth, by a QC code over $\mathbb{F}_q$ of index $\ell$ and co-index $m$ we mean either an $\mathcal{F}$-invariant subspace of $\mathbb{F}_q^{m\ell}$ or an $\mathbb{F}_q[x]$-submodule of $\oplus_{j=1}^{\ell}\mathbb{F}_q[x]$ that contains $\oplus_{j=1}^{\ell}\langle x^m-1\rangle$. We use these structures interchangeably, and the intended algebraic structure can be understood from the context without confusion. If $\mathcal{C}$ is a submodule of $\oplus_{j=1}^{\ell}\mathbb{F}_q[x]$ that contains $\oplus_{j=1}^{\ell}\langle x^m-1\rangle$, then $\mathcal{C}$ is free of rank $\ell$. This is because $\mathbb{F}_q[x]$ is a PID and both $\oplus_{j=1}^{\ell}\mathbb{F}_q[x]$ and $\oplus_{j=1}^{\ell}\langle x^m-1\rangle$ have a rank $\ell$. Therefore, any basis of $\mathcal{C}$ can be used to construct the rows of an $\ell\times\ell$ polynomial matrix over $\mathbb{F}_q[x]$. We call any such polynomial matrix a GPM of $\mathcal{C}$ and denote it by $\mathbf{G}$. For any GPM $\mathbf{G}$, there is an $\ell\times\ell$ polynomial matrix $\mathbf{A}$ such that
\begin{equation}
\label{identicalQC}
\mathbf{A}\mathbf{G}=\left( x^m-1\right)\mathbf{I}_\ell
\end{equation}
where $\mathbf{I}_\ell$ is the $\ell\times\ell$ identity matrix. The existence of $\mathbf{A}$ is immediate from the fact that $\mathcal{C}\supseteq \oplus_{j=1}^{\ell}\langle x^m-1\rangle$. The following formula relates the dimension of $\mathcal{C}$ as an $\mathbb{F}_q$-vector space to the polynomial matrix $\mathbf{A}$.
\begin{equation}
\label{Dim_QC}
\mathrm{dim}\left(\mathcal{C}\right)=\deg\left( \mathrm{det}\mathbf{A}\right)=m\ell -\deg\left( \mathrm{det}\mathbf{G}\right)
\end{equation} 
where $\deg$ and $\mathrm{det}$ stand respectively for degree and determinant.
 
\begin{theorem}
\label{GetGPM}
Let $B$ be an $r\times m\ell$ matrix over $\mathbb{F}_q$. Suppose that $\mathcal{C}$ is the smallest QC of index $\ell$ and co-index $m$ that contains the rows of $B$ as codewords. Let $\mathbf{M}$ be the polynomial matrix such that in block matrix notation
\begin{equation}
\mathbf{M}=\begin{pmatrix}
\mathfrak{B}_{r \times \ell}\\
(x^m-1)\mathbf{I}_\ell
\end{pmatrix}
\end{equation}
where $\mathrm{Row}_i\left( \mathfrak{B} \right)=\varphi\left(\mathrm{Row}_i\left( B\right) \right)$ for $1\le i\le r$. Here, $\mathrm{Row}_i$ denotes the $i$-th row. Then $\mathbf{M}$ generates $\mathcal{C}$ as a submodule of $\oplus_{j=1}^{\ell}\mathbb{F}_q[x]$.
\end{theorem}
\begin{proof}
Let $\mathcal{C}'$ be the $\mathbb{F}_q[x]$-submodule of $\oplus_{j=1}^{\ell}\mathbb{F}_q[x]$ generated by the rows of $\mathbf{M}$. Then $\oplus_{j=1}^{\ell}\langle x^m-1 \rangle \subseteq \mathcal{C}'$, that is to say, $\mathcal{C}'$ is QC of index $\ell$ and co-index $m$. But the rows of $B$ are codewords in $\mathcal{C}'$, so $\mathcal{C}\subseteq\mathcal{C}'$. Conversely, we know that $\oplus_{j=1}^{\ell}\langle x^m-1 \rangle \subseteq \mathcal{C}$. Consequently, the rows of $\mathbf{M}$ are codewords in $\mathcal{C}$. Hence $\mathcal{C}'\subseteq \mathcal{C}$.
\end{proof}
Using the same notation as in Theorem \ref{GetGPM}, the dimension of $\mathcal{C}$ can be determined by \eqref{Dim_QC} once a GPM of $\mathcal{C}$ is obtained. In fact, the reduced GPM of $\mathcal{C}$ can be determined by reducing $\mathbf{M}$ to its Hermite normal form. This is the content of Corollary \ref{GetGPM_Corr1}. However, Corollary \ref{GetGPM_Corr2} provides a formula for $\mathrm{dim}\left( \mathcal{C}\right)$ by evaluating the determinantal divisors of the polynomial matrix $\mathfrak{B}$. The $i$-th determinantal divisor of $\mathfrak{B}$, denoted $\mathfrak{d}_i\left(\mathfrak{B}\right)$, is the greatest common divisor (gcd) of all $i\times i$ determinantal minors of $\mathfrak{B}$. 
\begin{corollary}
\label{GetGPM_Corr1}
Given the same notation as Theorem \ref{GetGPM}, for any GPM $\mathbf{G}$ of $\mathcal{C}$, there exists an invertible polynomial matrix $\mathbf{U}$ such that
\begin{equation}
\label{eqinGetGPM_Corr1}
\mathbf{U}\mathbf{M}=\begin{pmatrix}\mathbf{G}_{\ell\times\ell}\\ \mathbf{0}_{r\times\ell}\end{pmatrix}
\end{equation}
where $\mathbf{0}_{r\times\ell}$ is the zero matrix of size $r\times\ell$. In particular, $\mathbf{G}$ is the reduced GPM of $\mathcal{C}$ if and only if $\mathbf{U}\mathbf{M}$ is the Hermite normal form of $\mathbf{M}$.
\end{corollary}
\begin{proof}
Suppose that $\mathbf{M}'$ is the Hermite normal form of $\mathbf{M}$ and suppose that $\mathbf{G}'$ is the reduced GPM of $\mathcal{C}$. There always exist invertible polynomial matrices $\mathbf{U}_1$ and $\mathbf{U}_2$ such that
\begin{equation*}
\mathbf{U}_1\mathbf{M}=\mathbf{M}'=\begin{pmatrix}\mathbf{G}'_{\ell\times\ell}\\ \mathbf{0}_{r\times\ell}\end{pmatrix}=\begin{pmatrix}\mathbf{U}_2\mathbf{G}_{\ell\times\ell}\\ \mathbf{0}_{r\times\ell}\end{pmatrix}=\begin{pmatrix}\mathbf{U}_2 & \mathbf{0}_{\ell\times\ell}\\ \mathbf{0}_{r\times\ell} & \mathbf{I}_{r}\end{pmatrix}\begin{pmatrix}\mathbf{G}_{\ell\times\ell}\\ \mathbf{0}_{r\times\ell}\end{pmatrix}.
\end{equation*}
The result then follows for
\begin{equation*}
\mathbf{U}=\begin{pmatrix}\mathbf{U}^{-1}_2 & \mathbf{0}_{\ell\times\ell}\\ \mathbf{0}_{r\times\ell} & \mathbf{I}_{r}\end{pmatrix}\mathbf{U}_1.
\end{equation*}
\end{proof}
\begin{corollary}
\label{GetGPM_Corr2}
With the notation of Theorem \ref{GetGPM}, 
\begin{equation*}
\mathrm{dim}\left( \mathcal{C}\right)=m\ell-\deg\left(\gcd_{0\le i\le \min\{r,\ell\}}\left\{\left(x^m-1\right)^{\ell-i}\mathfrak{d}_i\left(\mathfrak{B}\right)\right\}\right)
\end{equation*}
where $\mathfrak{d}_0\left(\mathfrak{B}\right)=1$.
\end{corollary}
\begin{proof}
Equation \eqref{Dim_QC} combined with Equation \eqref{eqinGetGPM_Corr1} shows that
\begin{equation*}
\mathrm{dim}\left( \mathcal{C}\right) 
=m\ell - \deg\left( \mathfrak{d}_\ell\left(\mathbf{G}\right)\right)=m\ell - \deg\left( \mathfrak{d}_\ell \begin{pmatrix}\mathbf{G}\\ \mathbf{0}_{r\times\ell}\end{pmatrix}\right)=m\ell - \deg\left( \mathfrak{d}_\ell\left(\mathbf{M}\right)\right).
\end{equation*}
But
\begin{equation*}
\begin{split}
\mathfrak{d}_\ell\left(\mathbf{M}\right)=\gcd\left\{\left(x^m-1\right)^\ell,\left(x^m-1\right)^{\ell-1}\mathfrak{d}_1\left(\mathfrak{B}\right), \left(x^m-1\right)^{\ell-2}\mathfrak{d}_2\left(\mathfrak{B}\right), \ldots, \left(x^m-1\right)^{\ell-w+1}\mathfrak{d}_{w-1}\left(\mathfrak{B}\right),\left(x^m-1\right)^{\ell-w}\mathfrak{d}_{w}\left(\mathfrak{B}\right) \right\}.
\end{split}
\end{equation*}
where $w=\min\{r,\ell\}$.
\end{proof}

QC codes find their broadest generalization in the class of MT codes. In \cite{Aydin2017}, MT codes have been introduced as an umbrella class for linear codes invariant under general constacyclic shifts. Precisely, these codes are defined as follows.
\begin{definition}
\label{def_MT}
Let $m_1,m_2,\ldots,m_\ell$ be positive integers and set $n=m_1+m_2+\cdots+m_\ell$. Let $\lambda_1,\lambda_2,\ldots,\lambda_\ell$ be non-zero elements of $\mathbb{F}_q$, set $\Lambda=\left(\lambda_1,\lambda_2,\ldots,\lambda_\ell \right)$, and call $\Lambda$ the shift constants. A $\Lambda$-MT code of index $\ell$ and block lengths $\left(m_1,m_2,\ldots,m_\ell\right)$ is a linear code over $\mathbb{F}_q$ of length $n$ invariant under the linear transformation
\begin{equation*}
\begin{split}
&\mathcal{T}_\Lambda :\left(a_{0,1},a_{1,1},\ldots,a_{m_1-1,1},a_{0,2},a_{1,2},\ldots,a_{m_2-1,2},\ldots,a_{0,\ell},a_{1,\ell},\ldots,a_{m_\ell-1,\ell}\right)\\
&\mapsto \left(\lambda_1 a_{m_1-1,1},a_{0,1},\ldots,a_{m_1-2,1},\lambda_2 a_{m_2-1,2},a_{0,2},\ldots,a_{m_2-2,2},\ldots, \lambda_\ell a_{m_\ell-1,\ell}, a_{0,\ell},\ldots,a_{m_\ell-2,\ell}\right)
\end{split}\end{equation*}
on $\mathbb{F}_q^n$.
\end{definition}
We now act in a very similar way to QC codes to equip MT codes with their polynomial representation. There is a vector space isomorphism between $\mathbb{F}_q^n$ and $\oplus_{j=1}^{\ell} \mathbb{F}_q[x]/\langle x^{m_j}-\lambda_j\rangle$ given by
\begin{equation}
\label{VS_Isomorphism_MT}
\begin{split}
\mathbf{a}=\left(a_{0,1},a_{1,1},\ldots,a_{m_1-1,1},a_{0,2},a_{1,2},\ldots,a_{m_2-1,2},\ldots,a_{0,\ell},a_{1,\ell},\ldots,a_{m_\ell-1,\ell}\right)\\
\mapsto\left(a_1\left(x\right)+\langle x^{m_1}-\lambda_1\rangle,a_2\left(x\right)+\langle x^{m_2}-\lambda_2\rangle,\ldots,a_\ell\left(x\right)+\langle x^{m_\ell}-\lambda_\ell\rangle\right)
\end{split}
\end{equation}
where $a_j\left(x\right)=\sum_{i=0}^{m_j-1}a_{i,j}x^i$ for $1\le j\le \ell$. The polynomial representation of $\mathbf{a}\in\mathbb{F}_q^n$ is the polynomial vector
\begin{equation}
\label{Poly_Vect}
\phi\left(\mathbf{a}\right)=\left( \sum_{i=0}^{m_1-1} a_{i,1}x^i , \sum_{i=0}^{m_2-1} a_{i,2}x^i, \ldots , \sum_{i=0}^{m_\ell-1} a_{i,\ell}x^i\right) \in \oplus_{j=1}^{\ell}\mathbb{F}_q[x].
\end{equation}
Briefly, $\mathbb{F}_q^n$ is an $\mathbb{F}_q[x]$-module with the action of $x$ regarded as $\mathcal{T}_\Lambda$. Then the map given by \eqref{VS_Isomorphism_MT} is an $\mathbb{F}_q[x]$-module isomorphism. Consequently, a MT code is an $\mathbb{F}_q[x]$-submodule of $\mathbb{F}_q^n$ (or $\oplus_{j=1}^{\ell} \mathbb{F}_q[x]/\langle x^{m_j}-\lambda_j\rangle$). But submodules of $\oplus_{j=1}^{\ell} \mathbb{F}_q[x]/\langle x^{m_j}-\lambda_j\rangle$ correspond to submodules of $\oplus_{j=1}^{\ell}\mathbb{F}_q[x]$ containing $\oplus_{j=1}^{\ell}\langle x^{m_j}-\lambda_j \rangle$. That is why we will often mean by a MT code $\mathcal{C}$ an $\mathbb{F}_q[x]$-submodule of $\oplus_{j=1}^{\ell}\mathbb{F}_q[x]$ that contains $\oplus_{j=1}^{\ell}\langle x^{m_j}-\lambda_j \rangle$. In this situation, $\mathcal{C}$ is free of rank $\ell$. Any basis of $\mathcal{C}$ can be used as rows to construct an $\ell\times\ell$ GPM of $\mathcal{C}$, denoted $\mathbf{G}$. Each MT code is assigned a unique GPM which is in the Hermite normal form, and is called the reduced GPM. For any GPM $\mathbf{G}$, since $\mathcal{C}\supseteq \oplus_{j=1}^{\ell}\langle x^{m_j}-\lambda_j \rangle$, there is a polynomial matrix $\mathbf{A}$ such that
\begin{equation}
\label{identicalMT}
\mathbf{A}\mathbf{G}=\mathrm{diag}\left[x^{m_j}-\lambda_j\right].
\end{equation}
By $\mathrm{diag}\left[x^{m_j}-\lambda_j\right]$ we mean the diagonal matrix with $x^{m_j}-\lambda_j$ in the $j$-th diagonal position. In other words,
\begin{equation*}
\begin{split}
\mathrm{diag}\left[x^{m_j}-\lambda_j\right]&=\mathrm{diag}\left[x^{m_1}-\lambda_1,\ldots,x^{m_\ell}-\lambda_\ell \right]=\begin{pmatrix}
x^{m_1}-\lambda_1 &0 &0&\cdots&0\\
0&x^{m_2}-\lambda_2  &0&\cdots&0\\
0&0&x^{m_3}-\lambda_3 &\cdots&0\\
\vdots&\vdots&\vdots&\ddots&\vdots\\
0&0&0&\cdots&x^{m_\ell}-\lambda_\ell
\end{pmatrix}.
\end{split}
\end{equation*}
Equation \eqref{identicalMT} is called the identical equation of $\mathbf{G}$. It is shown in \cite{Eldin1} that $\mathbf{A}$ determines the dimension of $\mathcal{C}$ by means of the equation 
\begin{equation}
\label{dimensionMT}
\mathrm{dim}\left(\mathcal{C}\right)=\deg\left(\mathrm{det}\mathbf{A}\right)=\sum_{j=1}^{\ell}m_j \ -\deg\left(\mathrm{det}\mathbf{G}\right).
\end{equation}

To summarize, a MT code over $\mathbb{F}_q$ of index $\ell$, block lengths $\left(m_1,m_2,\ldots,m_\ell\right)$, and shift constants $\Lambda=\left(\lambda_1,\lambda_2,\ldots,\lambda_\ell \right)$ can be viewed as a $\mathcal{T}_\Lambda$-invariant subspace of $\mathbb{F}_q^n$ or an $\mathbb{F}_q[x]$-submodule of $\oplus_{j=1}^{\ell}\mathbb{F}_q[x]$ that contains $\oplus_{j=1}^{\ell}\langle x^{m_j}-\lambda_j \rangle$. The intended algebraic structure can be understood from the context. We conclude this section by describing two other subclasses of MT codes that also generalize QC codes. The first is the class of QT codes. A $\lambda$-QT code is a MT code with equal block lengths and equal shift constants, i.e., $m_1=m_2=\cdots=m_\ell=m$ and $\lambda_1=\lambda_2=\cdots=\lambda_\ell=\lambda$. Particularly noteworthy is that a $\lambda$-constacyclic code is a QT code of index $\ell=1$. The second is the class of GQC codes. A GQC code is a MT code with $\lambda_1=\lambda_2=\cdots=\lambda_\ell=1$. Our final observation regarding these subclasses is that the algebraic structures described for MT codes apply to any of these codes.

\section{Galois duals and hulls}
\label{Sec3}
Hereinafter, we let $\mathbb{F}_q$ be a finite field of characteristic $p$ and dimension $e$ over its prime subfield, i.e., $q=p^e$. Denote by $\sigma$ the Frobenius automorphism of $\mathbb{F}_q$ given by $\sigma\left(\alpha \right)=\alpha^p$ for $\alpha\in\mathbb{F}_q$. The order of $\sigma$ is $e$, while the inverse of $\sigma^\kappa$ is $\sigma^{e-\kappa}$ for any $0\le \kappa<e$. We extend the definition of $\sigma$ in a natural way to maps on vectors, matrices, codes, polynomials, polynomial vectors, and polynomial matrices as follows.
\begin{definition}\begin{enumerate}
\item Let $\mathbf{a}\in\mathbb{F}_q^n$. Then for each $1\le i\le n$, $\mathrm{Ent}_{i}\left(\sigma \left(\mathbf{a}\right)\right)=\sigma\left(\mathrm{Ent}_{i}\left(\mathbf{a}\right)\right)$, where $\mathrm{Ent}_{i}\left(\mathbf{a}\right)$ is the $i$-th entry of $\mathbf{a}$. 
\item Let $S$ be a $k\times n$ matrix over $\mathbb{F}_q$. Then for $1\le i\le k$ and $1 \le j \le n$, $\mathrm{Ent}_{i,j}\left(\sigma \left(S \right)\right)=\sigma\left(\mathrm{Ent}_{i,j}\left(S\right)\right)$, where $\mathrm{Ent}_{i,j}\left(S\right)$ is the $(i, j)$-th entry of $S$.
\item Let $\mathcal{C}$ be a linear code over $\mathbb{F}_q$. Then $\sigma\mathcal{C}=\left\{\sigma\left(\mathbf{c}\right) : \mathbf{c}\in\mathcal{C}\right\}$.
\item Let $a(x)=\sum_{i=0}^{n-1}a_i x^i \in\mathbb{F}_q[x]$. Then $\sigma\left( a(x)\right)=\sum_{i=0}^{n-1}\sigma\left(a_i\right) x^i$.
\item Let $\mathbf{a}(x)\in \oplus_{j=1}^{\ell}\mathbb{F}_q[x]$. Then for $1\le j\le \ell$, $\mathrm{Ent}_{j}\left(\sigma \left(\mathbf{a}(x)\right)\right)=\sigma\left(\mathrm{Ent}_{j}\left(\mathbf{a}(x)\right)\right)$.
\item Let $\mathbf{S}$ be a polynomial matrix of order $r\times \ell$ over $\mathbb{F}_q[x]$. Then for $1\le i\le r$ and $1 \le j \le \ell$, $\mathrm{Ent}_{i,j}\left(\sigma \left(\mathbf{S} \right)\right)=\sigma\left(\mathrm{Ent}_{i,j}\left(\mathbf{S}\right)\right)$.
\end{enumerate}
\end{definition}

The Euclidean inner product on $\mathbb{F}_q^n$ is defined by
\begin{equation*}
\langle \mathbf{a},\mathbf{b}\rangle=\sum_{i=1}^{n} \mathrm{Ent}_{i}\left(\mathbf{a}\right)\mathrm{Ent}_{i}\left(\mathbf{b}\right) \quad \text{for all } \mathbf{a},\mathbf{b}\in\mathbb{F}_q^n.
\end{equation*}
It follows at once that $\langle \cdot,\cdot \rangle$ is a bilinear symmetric form. Indeed, suppose that $\mathbf{a},\mathbf{b},\mathbf{c}\in\mathbb{F}_q^n$ and $\alpha,\beta\in\mathbb{F}_q$. Then
\begin{enumerate}
\item $\langle \mathbf{a},\mathbf{b}\rangle=\langle \mathbf{b},\mathbf{a}\rangle$,
\item $\langle \alpha\mathbf{a}+\beta\mathbf{b},\mathbf{c}\rangle=\alpha \langle \mathbf{a},\mathbf{c}\rangle+\beta \langle \mathbf{b},\mathbf{c}\rangle$,\text{ and}
\item $\langle \mathbf{a},\beta\mathbf{b}+\alpha\mathbf{c}\rangle=\beta \langle \mathbf{a},\mathbf{b}\rangle+\alpha \langle \mathbf{a},\mathbf{c}\rangle$.
\end{enumerate}
For any $0\le\kappa<e$, the $\kappa$-Galois inner product was introduced in \cite{Fan2016} as a generalization of the Euclidean inner product; it is denoted $\langle \cdot,\cdot \rangle_\kappa$ and defined as follows.
\begin{definition}
\label{Def_Galois_inner}
Let $\kappa$ be an integer such that $0\le \kappa<e$. The $\kappa$-Galois inner product of $\mathbf{a},\mathbf{b}\in\mathbb{F}_q^n$ is defined to be 
\begin{equation*}
\langle \mathbf{a},\mathbf{b}\rangle_\kappa=\langle \mathbf{a},\sigma^\kappa\left(\mathbf{b}\right)\rangle.
\end{equation*}
\end{definition}
For any $\mathbf{a},\mathbf{b},\mathbf{c}\in\mathbb{F}_q^n$ and $\alpha,\beta\in\mathbb{F}_q$, one can easily check the following facts:
\begin{enumerate}
\item $\langle \mathbf{a},\mathbf{b}\rangle_\kappa=\sigma^\kappa\left(\langle \mathbf{b},\mathbf{a}\rangle_{e-\kappa}\right)$,
\item $\langle \alpha\mathbf{a}+\beta\mathbf{b},\mathbf{c}\rangle_\kappa=\alpha \langle \mathbf{a},\mathbf{c}\rangle_\kappa+\beta \langle \mathbf{b},\mathbf{c}\rangle_\kappa$, and
\item $\langle \mathbf{a},\beta\mathbf{b}+\alpha\mathbf{c}\rangle_\kappa=\sigma^\kappa\left(\beta\right) \langle \mathbf{a},\mathbf{b}\rangle_\kappa+\sigma^\kappa\left(\alpha\right) \langle \mathbf{a},\mathbf{c}\rangle_\kappa$.
\end{enumerate}

Recall that the Euclidean dual of a linear code $\mathcal{C}$ over $\mathbb{F}_q$ of length $n$ is given by
\begin{equation*}
\mathcal{C}^\perp=\left\{ \mathbf{a}\in\mathbb{F}_q^n \text{ such that } \langle \mathbf{c},\mathbf{a}\rangle =0 \ \forall \ \mathbf{c}\in\mathcal{C}\right\}.
\end{equation*}
The $\kappa$-Galois dual of $\mathcal{C}$, denoted $\mathcal{C}^{\perp_\kappa}$, is defined analogously where the Euclidean inner product is replaced by the $\kappa$-Galois inner product. That is,
\begin{equation*}
\mathcal{C}^{\perp_\kappa}=\left\{ \mathbf{a}\in\mathbb{F}_q^n \text{ such that } \langle \mathbf{c},\mathbf{a}\rangle_\kappa =0 \ \forall \ \mathbf{c}\in\mathcal{C}\right\}.
\end{equation*}
Note that if $\mathcal{C}$ is linear of dimension $k$, then both $\mathcal{C}^\perp$ and $\mathcal{C}^{\perp_\kappa}$ are linear of dimension $n-k$. More precisely, $\mathcal{C}^{\perp_\kappa}=\sigma^{e-\kappa}\left(\mathcal{C}^\perp\right)$ and $\mathcal{C}^\perp=\sigma^{\kappa}\left(\mathcal{C}^{\perp_\kappa}\right)$. This is because $\langle\mathbf{a},\mathbf{b}\rangle=\langle\mathbf{a},\sigma^{e-\kappa}\left(\mathbf{b}\right)\rangle_\kappa$ by Definition \ref{Def_Galois_inner}. Furthermore, it is known that
$$\left(\mathcal{C}^{\perp}\right)^\perp=\mathcal{C} \quad \text{and} \quad \sigma\left(\mathcal{C}^\perp\right)=\left(\sigma \mathcal{C}\right)^\perp.$$
Then it follows that
$$\left(\mathcal{C}^{\perp_{\kappa_1}}\right)^{\perp_{\kappa_2}}=\sigma^{e-\kappa_2}\left(\mathcal{C}^{\perp_{\kappa_1}}\right)^\perp=\sigma^{e-\kappa_2}\left(\sigma^{e-\kappa_1}\mathcal{C}^\perp\right)^\perp=\sigma^{2e-\kappa_1-\kappa_2}\left(\mathcal{C}\right).$$

Linear codes can be classified according to their intersection with their duals. A linear code either trivially intersects its dual, is contained in its dual, or non-trivially intersects its dual. To clarify these cases precisely, it is convenient to make the following definition.
\begin{definition}
\label{def_LCD_SO}
Let $q=p^e$ and $0\le\kappa<e$. Let $\mathcal{C}$ be a linear code over $\mathbb{F}_q$.
\begin{enumerate}
\item $\mathcal{C}$ is Euclidean (respectively, $\kappa$-Galois) LCD if $\mathcal{C}\cap \mathcal{C}^\perp=\{\mathbf{0}\}$ (respectively, $\mathcal{C}\cap \mathcal{C}^{\perp_\kappa}=\{\mathbf{0}\}$).
\item $\mathcal{C}$ is Euclidean (respectively, $\kappa$-Galois) self-orthogonal if $\mathcal{C}\subseteq \mathcal{C}^\perp$ (respectively, $\mathcal{C}\subseteq \mathcal{C}^{\perp_\kappa}$).
\item The Euclidean (respectively, $\kappa$-Galois) hull of $\mathcal{C}$ is $h\left(\mathcal{C}\right)=\mathcal{C}\cap \mathcal{C}^\perp$ (respectively, $h_\kappa\left(\mathcal{C}\right)=\mathcal{C}\cap \mathcal{C}^{\perp_\kappa}$).
\end{enumerate}
\end{definition}
It is worth pointing out explicitly that $h_\kappa\left(\mathcal{C}\right)=\left\{\mathbf{0}\right\}$ if and only if $\mathcal{C}$ is $\kappa$-Galois LCD, whereas $h_\kappa\left(\mathcal{C}\right)=\mathcal{C}$ if and only if $\mathcal{C}$ is $\kappa$-Galois self-orthogonal. Nevertheless, it might be useful to investigate the case in which $\mathbf{0}\subsetneq h_\kappa\left(\mathcal{C}\right) \subsetneq \mathcal{C}$. For a linear code $\mathcal{C}$ over $\mathbb{F}_q$, the dimension of $h_\kappa\left(\mathcal{C}\right)$ can be determined by the following result which is proven in \cite{Hongwei}.

\begin{theorem}
\label{hongwei}
Let $q=p^e$ and $0\le\kappa<e$. Let $S$ be a generator matrix of a linear code $\mathcal{C}$ over $\mathbb{F}_q$. Then
$$\mathrm{Rank}\left(S \left( \sigma^\kappa S \right)^t \right) = \mathrm{dim}\left(\mathcal{C}\right)-\mathrm{dim}\left(h_\kappa\left(\mathcal{C}\right)\right)$$
where $^t$ denotes the transpose. 
\end{theorem}
In the special cases of $\kappa$-Galois LCD and self-orthogonal codes, Theorem \ref{hongwei} can be stated as follows: A linear code $\mathcal{C}$ with a generator matrix $S$ is $\kappa$-Galois LCD (respectively, self-orthogonal) if and only if the rank of $S \left( \sigma^\kappa S \right)^t$ is equal to $\mathrm{dim}\left(\mathcal{C}\right)$ (respectively, $0$). The first aim of Section \ref{main_results_section} below is to provide a formula for the Galois hull dimension of a MT code in terms of its GPM. Furthermore, a GPM for the Galois hull is constructed. To this end, we shall now examine the duals of MT codes.

\begin{theorem}
\label{dualMT}
Let $0\le\kappa<e$ and let $\Lambda=\left(\lambda_1,\lambda_2,\ldots,\lambda_\ell\right)$, where $0\ne\lambda_j\in\mathbb{F}_q$ for $1\le j\le \ell$. If $\mathcal{C}$ is a $\Lambda$-MT code over $\mathbb{F}_q$ of index $\ell$ and block lengths $\left(m_1,m_2,\ldots,m_\ell\right)$, then $\mathcal{C}^{\perp_\kappa}$ is $\Delta_\kappa$-MT with block lengths $\left(m_1,m_2,\ldots,m_\ell\right)$, where
$$\Delta_\kappa=\sigma^{e-\kappa}\left(\frac{1}{\lambda_1},\frac{1}{\lambda_2},\ldots,\frac{1}{\lambda_\ell}\right).$$
\end{theorem}
\begin{proof}
Suppose $\mathbf{a}\in\mathcal{C}^{\perp_\kappa}$. Then $0=\langle \mathcal{T}_\Lambda^{-1}\left(\mathbf{c}\right),\mathbf{a}\rangle_\kappa=\langle \mathbf{c},\mathcal{T}_{\Delta_\kappa}\left(\mathbf{a}\right)\rangle_\kappa$ for all $\mathbf{c}\in\mathcal{C}$. This implies that $\mathcal{T}_{\Delta_\kappa}\left(\mathbf{a}\right)\in\mathcal{C}^{\perp_\kappa}$. Therefore, $\mathcal{C}^{\perp_\kappa}$ is $\mathcal{T}_{\Delta_\kappa}$-invariant. 
\end{proof}
Since the dual of a MT code was found by Theorem \ref{dualMT} to be MT as well, one might expect a formula for a GPM of this dual. We use $\mathbf{H}_\kappa$ (respectively, $\mathbf{H}$) to denote a GPM for $\mathcal{C}^{\perp_\kappa}$ (respectively, $\mathcal{C}^\perp$). It is evident that $\mathbf{H}_\kappa=\sigma^{e-\kappa}\left(\mathbf{H}\right)$ because $\mathcal{C}^{\perp_\kappa}=\sigma^{e-\kappa}\left(\mathcal{C}^\perp\right)$. Then it suffices to introduce a GPM formula for $\mathcal{C}^\perp$. This is exactly what the following theorem does, the proof of which can be found in \cite{Eldin2}.

\begin{theorem}
\label{GPMforDual}
Let $0\le\kappa<e$ and let $\Lambda=\left(\lambda_1,\lambda_2,\ldots,\lambda_\ell\right)$, where $0\ne\lambda_j\in\mathbb{F}_q$ for $1\le j\le \ell$. Let $\mathbf{G}$ be the reduced GPM of a $\Lambda$-MT code $\mathcal{C}$ over $\mathbb{F}_q$ of index $\ell$ and block lengths $\left(m_1,m_2,\ldots,m_\ell\right)$. Suppose $\mathbf{A}$ is the polynomial matrix that satisfies the identical equation \eqref{identicalMT} of $\mathbf{G}$. Construct a polynomial matrix $\mathbf{H}$ such that 
\begin{equation*}
\mathrm{Col}_{i}\left(\mathbf{H}\right)\equiv\mathrm{Row}_{i}\left(\mathrm{diag}\left[x^{m_j}\right] \mathbf{A}\left(\frac{1}{x}\right) \mathrm{diag}\left[x^{-d_j}\right]\right) \pmod{x^{m_i}-\frac{1}{\lambda_i}}
\end{equation*}
for $1\le i\le \ell$. Here, $\mathrm{Col}_{i}$ denotes the $i$-th column, $\mathrm{Row}_{i}$ denotes the $i$-th row, $d_j=\deg\left(\mathrm{Ent}_{j,j}\left(\mathbf{G}\right)\right)$ for $1\le j\le \ell$, and $\mathbf{A}\left(\frac{1}{x}\right)$ is the matrix obtained from $\mathbf{A}$ by replacing $x$ by $\frac{1}{x}$. Then $\mathbf{H}$ is a GPM for $\mathcal{C}^\perp$, while $\mathbf{H}_\kappa=\sigma^{e-\kappa}\left(\mathbf{H}\right)$ is a GPM for $\mathcal{C}^{\perp_\kappa}$.
\end{theorem}
We shall now consider a sufficient condition under which the Galois hull of a MT code is MT as well.
\begin{corollary}
\label{coro_MT_hull}
Let $\Lambda=\left(\lambda_1,\lambda_2,\ldots,\lambda_\ell\right)$ and let $\mathcal{C}$ be a $\Lambda$-MT code over $\mathbb{F}_q$ of index $\ell$ and block lengths $\left(m_1,m_2,\ldots,m_\ell\right)$. If $\lambda_j^{-p^{e-\kappa}}=\lambda_j$ for $1\le j\le \ell$, then $h_\kappa\left(\mathcal{C}\right)$ is $\Lambda$-MT code over $\mathbb{F}_q$ of index $\ell$ and block lengths $\left(m_1,m_2,\ldots,m_\ell\right)$.
\end{corollary}
\begin{proof}
If $\lambda_j^{-p^{e-\kappa}}=\lambda_j$, then $\Delta_\kappa=\Lambda$. Then $\mathcal{C}^{\perp_\kappa}$, and hence $\mathcal{C}\cap\mathcal{C}^{\perp_\kappa}$, is $\Lambda$-MT by Theorem \ref{dualMT}.
\end{proof}

\begin{example}
\label{Ex_SO_GQC}
Let $p=2$, $e=4$, $q=16$, $n=9$, $\ell=3$, $m_1=3$, $m_2=2$, $m_3=4$, and $\lambda_1=\lambda_2=\lambda_3=1$. Consider the GQC code $\mathcal{C}$ over $\mathbb{F}_q$ of index $\ell$, block lengths $\left(m_1,m_2,m_3\right)$, and reduced GPM
\begin{equation}
\label{Ex_GQC}
\mathbf{G}=\begin{pmatrix}
x^2 + x + 1 &  1 &  \theta^5 x + \theta^5 \\
0 &  x + 1  & \theta^5 x^2 + \theta^5 \\
0 &  0 &  x^3 + x^2 + x + 1  
\end{pmatrix}\end{equation}
where $\theta\in\mathbb{F}_q$ such that $\theta^4+\theta+1=0$. The matrix that satisfies the identical equation \eqref{identicalMT} of $\mathbf{G}$ is
\begin{equation*}
\mathbf{A}=\begin{pmatrix}
x + 1 & 1 & 0 \\
0 & x + 1 & \theta^{5} \\
0 & 0 & x + 1
\end{pmatrix}.
\end{equation*}
The construction of a GPM $\mathbf{H}$ for the Euclidean dual $\mathcal{C}^\perp$ (via an application of Theorem \ref{GPMforDual}) gives  
\begin{equation*}
\mathrm{Col}_{i}\left(\mathbf{H}\right)\equiv\mathrm{Row}_{i}\left(\mathrm{diag}\left[x^{3},x^{2},x^{4}\right] \mathbf{A}\left(\frac{1}{x}\right) \mathrm{diag}\left[x^{-2},x^{-1},x^{-3}\right]\right) \pmod{x^{m_i}-1}
\end{equation*}
for $i=1,2,3$. That is,
\begin{equation*}
\mathbf{H}=\begin{pmatrix}
 x+1 & 0 & 0 \\
x^2 &  x+1 & 0 \\
0 & \theta^{5}x &  x+1
\end{pmatrix}.
\end{equation*}
If we let $\kappa=3$, then, by Theorem \ref{dualMT}, the $\kappa$-Galois dual $\mathcal{C}^{\perp_3}$ of $\mathcal{C}$ is GQC as well. By using Theorem \ref{GPMforDual} again, a GPM for $\mathcal{C}^{\perp_3}$ is
\begin{equation*}
\mathbf{H}_3=\sigma\left(\mathbf{H}\right)=\begin{pmatrix}
 x+1 & 0 & 0 \\
x^2 &  x+1 & 0 \\
0 & \theta^{10}x &  x+1
\end{pmatrix}.
\end{equation*}
We claim that $\mathcal{C}$ is $\kappa$-Galois self-orthogonal. In fact, $\mathcal{C}\subseteq\mathcal{C}^{\perp_3}$ if and only if $\mathbf{G}=\mathbf{M}\mathbf{H}_3$ for some polynomial matrix $\mathbf{M}$. Our claim is true because $\mathbf{G}=\mathbf{M}\mathbf{H}_3$ for
\begin{equation*}
\mathbf{M}=\begin{pmatrix}
1 & 1 & \theta^{5} \\
x^{2} & x + 1 & \theta^{5} x + \theta^{5} \\
\theta^{10} x^{3} & \theta^{10} x^{2} + \theta^{10} x & x^{2} + 1
\end{pmatrix}.
\end{equation*}
\end{example}

\section{A GPM of $h_\kappa\left(\mathcal{C}\right)$}
\label{main_results_section}
Our first goal in this section is to take advantage of Theorem \ref{hongwei} and provide a formula for the dimension of $h_\kappa\left(\mathcal{C}\right)$ from a given GPM $\mathbf{G}$ of a MT code $\mathcal{C}$. Then the second goal is to get a deeper result than simply applying Theorem \ref{hongwei} to the specific case of MT codes, but deducing a GPM for $h_\kappa\left(\mathcal{C}\right)$. In view of Corollary \ref{coro_MT_hull}, $h_\kappa\left(\mathcal{C}\right)$ for a $\Lambda$-MT code $\mathcal{C}$ is $\Lambda$-MT if $\lambda_j^{-p^{-\kappa}}=\lambda_j$ for $1\le j\le \ell$. Throughout this section, we assume that these conditions remain satisfied and we write this as $\Lambda^{-p^{-\kappa}}=\Lambda$. We further let $\mathcal{T}_\Lambda$ denote the linear transformation in Definition \ref{def_MT} and let $N$ denote the order of $\mathcal{T}_\Lambda$, i.e., $\mathcal{T}_\Lambda^N$ is the identity map. It is evident that if $\Lambda=\left(\lambda_1,\lambda_2,\ldots,\lambda_\ell\right)$ and the block lengths are $\left(m_1,m_2,\ldots,m_\ell\right)$, then $N=\mathrm{lcm}\left\{t_1 m_1,t_2 m_2,\ldots,t_\ell m_\ell\right\}$, where $t_j$ is the multiplicative order of $\lambda_j$ for $1\le j\le \ell$ and $\mathrm{lcm}$ stands for the least common multiple. We start with a few preliminaries before establishing the main findings.

\begin{lemma}
\label{lem_shift}
Assume that $\Lambda^{-p^{-\kappa}}=\Lambda$. Let $q=p^e$ and $0\le\kappa<e$. Then for any $\mathbf{a},\mathbf{b}\in\mathbb{F}_q^n$,
\begin{equation*}
\langle \mathcal{T}_\Lambda\left(\mathbf{a}\right),\mathbf{b}\rangle_\kappa=\langle \mathbf{a},\mathcal{T}_\Lambda^{N-1}\left(\mathbf{b}\right)\rangle_\kappa.
\end{equation*}
\end{lemma}
\begin{proof}
Subdivide each of the vectors $\mathbf{a}$ and $\mathbf{b}$ into $\ell$ blocks of lengths $m_j$ as in Definition  \ref{def_MT}. Then
\begin{equation*}
\begin{split}
\langle \mathcal{T}_\Lambda\left(\mathbf{a}\right),\mathbf{b}\rangle_\kappa&=\sum_{j=1}^{\ell}\left( \lambda_j a_{m_j-1,j}\sigma^\kappa\left(b_{0,j}\right) + \sum_{u=1}^{m_j-1} a_{u-1,j}\sigma^\kappa\left(b_{u,j}\right)\right)\\
&=\sum_{j=1}^{\ell}\left(  a_{m_j-1,j}\sigma^\kappa\left(\lambda_j^{p^{-\kappa}} b_{0,j}\right) + \sum_{u=1}^{m_j-1} a_{u-1,j}\sigma^\kappa\left(b_{u,j}\right)\right)\\
&=\sum_{j=1}^{\ell}\left(  a_{m_j-1,j}\sigma^\kappa\left(\lambda_j^{-1} b_{0,j}\right) + \sum_{u=1}^{m_j-1} a_{u-1,j}\sigma^\kappa\left(b_{u,j}\right)\right) \\
&=\langle \mathbf{a},\mathcal{T}_\Lambda^{N-1}\left(\mathbf{b}\right)\rangle_\kappa
\end{split}
\end{equation*}
where we used that $\lambda_j^{p^{-\kappa}}=\lambda_j^{-1}$ for $1\le j \le \ell$.
\end{proof}

From \eqref{Poly_Vect}, any vector $\mathbf{a}\in\mathbb{F}_q^n$ is represented by the polynomial vector $\phi\left(\mathbf{a}\right)$. By $\phi\left(\mathbf{a}\right)\left(\frac{1}{x} \right)$, we shall mean to replace $x$ by $\frac{1}{x}$ in $\phi\left(\mathbf{a}\right)$. The following important result expresses the $\kappa$-Galois inner product in terms of the polynomial vector representation.
\begin{lemma}
\label{lem1}
Assume that $\Lambda^{-p^{-\kappa}}=\Lambda$. Then for any $\mathbf{a},\mathbf{b}\in\mathbb{F}_q^n$,
\begin{equation*}
\phi\left( \mathbf{a}\right)\mathrm{diag}\left[ \frac{x^N-1}{x^{m_j}-\lambda_j}\right]  \sigma^\kappa\left(\phi\left(\mathbf{b}\right)\left(\frac{1}{x}\right) \mathrm{diag}\left[ x^{m_j}\right] \right)^t \equiv \sum_{\iota=0}^{N-1}\langle\mathbf{a},\mathcal{T}^\iota_\Lambda\left(\mathbf{b}\right) \rangle_\kappa \ x^\iota \pmod{\left(x^N-1\right)}.
\end{equation*} 
\end{lemma}
\begin{proof}
From \eqref{Poly_Vect} we see that $\phi\left( \mathbf{a}\right)=\left(a_1(x),a_2(x),\ldots,a_\ell(x) \right)$, where $a_j(x)=\sum_{u=0}^{m_j-1} a_{u,j}x^u$ for $1\le j\le \ell$. However, for any $1\le j\le \ell$, we note that 
\begin{equation*}
\begin{split}
\frac{x^N-1}{x^{m_j}-\lambda_j}&=\frac{\left(x^{m_j}\right)^{N/m_j}-\left(\lambda_j \right)^{N/m_j}}{x^{m_j}-\lambda_j}=x^{N-m_j}+\lambda_j x^{N-2m_j}+\lambda_j^2 x^{N-3m_j}+\cdots+ \lambda_j^{N/m_j-2} x^{m_j}+ \lambda_j^{N/m_j-1}\\
&=\sum_{s=1}^{N/m_j}\lambda_j^{s-1} x^{N-s m_j}.
\end{split}
\end{equation*}
Again, if $\phi\left( \mathbf{b}\right)=\left(b_1(x),b_2(x),\ldots,b_\ell(x) \right)$ and $b_j(x)=\sum_{u=0}^{m_j-1} b_{u,j}x^u$, then the $j$-th entry of $\sigma^\kappa\left(\phi\left(\mathbf{b}\right)\left(\frac{1}{x}\right) \mathrm{diag}\left[ x^{m_j}\right]\right)$ is
\begin{equation*}
\sum_{u=0}^{m_j-1} \sigma^\kappa\left(b_{u,j}\right)x^{m_j-u}=\sum_{v=1}^{m_j} \sigma^\kappa\left(b_{m_j-v,j}\right)x^{v}.
\end{equation*}
Using the assumption on $\Lambda$, it follows that $\lambda_j= \sigma^\kappa\left(\lambda_j^{p^{-\kappa}} \right)= \sigma^\kappa\left(\lambda_j^{-1} \right)$ for $1\le j\le \ell$. Therefore,
\begin{equation*}
\begin{split}
&\phi\left( \mathbf{a}\right)\mathrm{diag}\left[ \frac{x^N-1}{x^{m_j}-\lambda_j}\right]  \sigma^\kappa\left(\phi\left(\mathbf{b}\right)\left(\frac{1}{x}\right) \mathrm{diag}\left[ x^{m_j}\right] \right)^t \\
&=\sum_{j=1}^{\ell} \sum_{u=0}^{m_j-1} a_{u,j}x^u  \sum_{s=1}^{N/m_j}\lambda_j^{s-1} x^{N-s m_j} \sum_{v=1}^{m_j} \sigma^\kappa\left(b_{m_j-v,j}\right)x^{v}\\
&=\sum_{j=1}^{\ell} \sum_{h=1}^{2m_j-1} \mathop{\sum_{u=0}^{m_j-1} \sum_{v=1}^{m_j}}_{u+v=h} \sum_{s=1}^{N/m_j} \lambda_j^{s-1} a_{u,j}\sigma^\kappa\left(b_{m_j-v,j}\right) x^{N-s m_j+h}\\
&=\sum_{j=1}^{\ell} \sum_{h=1}^{2m_j-1} \mathop{\sum_{u=0}^{m_j-1} \sum_{v=1}^{m_j}}_{u+v=h} \sum_{s=1}^{N/m_j}  a_{u,j}\sigma^\kappa\left(\lambda_j^{-(s-1)} b_{m_j-v,j}\right) x^{N-s m_j+h}\\
&= \sum_{j=1}^{\ell} \sum_{h=1}^{2m_j-1} \mathop{\sum_{u=0}^{m_j-1} \sum_{v=1}^{m_j}}_{u+v=h} \sum_{s=2}^{N/m_j}  a_{u,j}\sigma^\kappa\left(\lambda_j^{-(s-1)} b_{m_j-v,j}\right) x^{N-s m_j+h} +\sum_{j=1}^{\ell} \sum_{h=1}^{2m_j-1} \mathop{\sum_{u=0}^{m_j-1} \sum_{v=1}^{m_j}}_{u+v=h}  a_{u,j}\sigma^\kappa\left( b_{m_j-v,j}\right) x^{N- m_j+h}.
\end{split}\end{equation*}
For each $1\le \iota \le N$ and $1\le j\le \ell$, write $\iota=q_j m_j - r_j$ for some integers $q_j$ and $0\le r_j \le m_j-1$. Then, the reduction of the above equation modulo $x^N-1$ yields
\begin{equation*}
\begin{split}
&\phi\left( \mathbf{a}\right)\mathrm{diag}\left[ \frac{x^N-1}{x^{m_j}-\lambda_j}\right]  \sigma^\kappa\left(\phi\left(\mathbf{b}\right)\left(\frac{1}{x}\right) \mathrm{diag}\left[ x^{m_j}\right] \right)^t \\
&\equiv \sum_{j=1}^{\ell} \sum_{h=1}^{2m_j-1} \mathop{\sum_{u=0}^{m_j-1} \sum_{v=1}^{m_j}}_{u+v=h} \sum_{s=2}^{N/m_j}  a_{u,j}\sigma^\kappa\left(\lambda_j^{-(s-1)} b_{m_j-v,j}\right) x^{N-s m_j+h}\\
&\qquad+\sum_{j=1}^{\ell} \sum_{h=1}^{m_j-1} \mathop{\sum_{u=0}^{m_j-1} \sum_{v=1}^{m_j}}_{u+v=h}  a_{u,j}\sigma^\kappa\left( b_{m_j-v,j}\right) x^{N- m_j+h}\\
&\qquad+\sum_{j=1}^{\ell} \sum_{h=m_j}^{2m_j-1} \mathop{\sum_{u=0}^{m_j-1} \sum_{v=1}^{m_j}}_{u+v=h}  a_{u,j}\sigma^\kappa\left( b_{m_j-v,j}\right) x^{h- m_j} \\
&\equiv \sum_{j=1}^{\ell} \sum_{\iota=1}^{N} \left(  \sum_{u=0}^{r_j-1} a_{u,j}\sigma^\kappa\left(\lambda_j^{-(q_j-1)} b_{m_j-r_j+u,j}\right)   + \sum_{u=r_j}^{m_j-1} a_{u,j}\sigma^\kappa\left(\lambda_j^{-q_j} b_{u-r_j,j}\right) \right) x^{N-\iota}\\
&=\sum_{\iota=1}^{N} \langle \mathbf{a},\mathcal{T}_\Lambda^{N-\iota}\left(\mathbf{b}\right)\rangle_\kappa \ x^{N-\iota}=\sum_{\iota=0}^{N-1}\langle\mathbf{a},\mathcal{T}^\iota_\Lambda\left(\mathbf{b}\right) \rangle_\kappa \ x^\iota.
\end{split}\end{equation*} 
\end{proof}

An example of the usefulness of Lemma \ref{lem1} is to characterize a GPM for the $\kappa$-Galois dual of a MT code as follows.
\begin{corollary}
\label{corollary1}
Let $\mathcal{C}$ be a $\Lambda$-MT code with a GPM $\mathbf{G}$. If $\mathbf{H}_\kappa$ is a GPM of $\mathcal{C}^{\perp_\kappa}$, then
\begin{equation*}
\mathbf{G}\ \mathrm{diag}\left[ \frac{x^N-1}{x^{m_j}-\lambda_j}\right]  \sigma^\kappa\left(\mathbf{H}_\kappa \left(\frac{1}{x}\right) \mathrm{diag}\left[ x^{m_j}\right] \right)^t \equiv \mathbf{0} \pmod{\left(x^N-1\right)}
\end{equation*} 
where by $\mathbf{H}_\kappa \left(\frac{1}{x}\right)$ we mean to replace $x$ by $\frac{1}{x}$ in $\mathbf{H}_\kappa$.
\end{corollary}
\begin{proof}
For any pair $(i,j)$ with $1\le i,j\le \ell$, let $\mathbf{a}$ and $\mathbf{b}$ be the codewords of $\mathcal{C}$ and $\mathcal{C}^{\perp_\kappa}$ that correspond to the $i$-th and $j$-th rows of $\mathbf{G}$ and $\mathbf{H}_\kappa $, respectively. Then there exist polynomial vectors $\mathbf{u}(x)$ and $\mathbf{v}(x)$ for which \begin{equation*}
\mathrm{Row}_i\left(\mathbf{G}\right)=\phi\left(\mathbf{a}\right)+\mathbf{u}(x) \ \mathrm{diag}\left[x^{m_j}-\lambda_j\right]
\end{equation*}
and
\begin{equation*}
\mathrm{Row}_j\left(\mathbf{H}_\kappa\right)=\phi\left(\mathbf{b}\right)+\mathbf{v}(x) \ \mathrm{diag}\left[x^{m_j}-\lambda_j\right].
\end{equation*}
Therefore we have 
\begin{equation*}
\begin{split}
&\mathrm{Ent}_{i,j}\left(\mathbf{G}\ \mathrm{diag}\left[ \frac{x^N-1}{x^{m_j}-\lambda_j}\right] \sigma^\kappa\left(\mathbf{H}_\kappa \left(\frac{1}{x}\right) \mathrm{diag}\left[ x^{m_j}\right] \right)^t\right)\\
& = \mathrm{Row}_{i}\left(\mathbf{G}\right) \mathrm{diag}\left[ \frac{x^N-1}{x^{m_j}-\lambda_j}\right]  \sigma^\kappa\left(\mathrm{Row}_{j}\mathbf{H}_\kappa \left(\frac{1}{x}\right) \mathrm{diag}\left[ x^{m_j}\right] \right)^t \\
& \equiv \phi\left(\mathbf{a}\right) \mathrm{diag}\left[ \frac{x^N-1}{x^{m_j}-\lambda_j}\right]  \sigma^\kappa\left(\phi\left(\mathbf{b}\right) \left(\frac{1}{x}\right) \mathrm{diag}\left[ x^{m_j}\right] \right)^t \pmod{\left(x^N-1\right)}\\
& \equiv \sum_{\iota=0}^{N-1}\langle\mathbf{a},\mathcal{T}^\iota_\Lambda\left(\mathbf{b}\right) \rangle_\kappa \ x^\iota \pmod{\left(x^N-1\right)}. \qquad \text{(By Lemma \ref{lem1})}
\end{split}
\end{equation*} 
Under the assumption that $\Lambda^{-p^{-\kappa}}=\Lambda$, $\mathcal{C}^{\perp_\kappa}$ is $\Lambda$-MT. Thus, $\mathcal{T}_\Lambda^{\iota}\left(\mathbf{b}\right) \in \mathcal{C}^{\perp_\kappa}$ and $\langle \mathbf{a},\mathcal{T}_\Lambda^{\iota}\left(\mathbf{b}\right)\rangle_\kappa =0$ for every $0\le \iota\le N-1$.
\end{proof}

To each GPM $\mathbf{G}$ of a MT code we will associate a QC code of length $N$ and index $\ell$. In order to precisely define this QC code, we first need to define a polynomial matrix associated with $\mathbf{G}$. 
\begin{definition}
\label{Def_of_B}
Let $\mathbf{G}$ be a GPM of a $\Lambda$-MT code, where $\Lambda^{-p^{-\kappa}}=\Lambda$. We define 
\begin{equation*}
\mathfrak{B}_\mathbf{G}\equiv\mathbf{G}\ \mathrm{diag}\left[ \frac{x^N-1}{x^{m_j}-\lambda_j}\right] \sigma^\kappa\left(\mathbf{G}\left(\frac{1}{x}\right)\mathrm{diag}\left[ x^{m_j}\right] \right)^t \pmod{\left(x^N-1\right)}
\end{equation*}
such that entries of $\mathfrak{B}_\mathbf{G}$ are of degree at most $N-1$. Again, by $\mathbf{G} \left(\frac{1}{x}\right)$ we mean to replace $x$ by $\frac{1}{x}$ in $\mathbf{G}$.
\end{definition}
We will now use Lemma \ref{lem1} to show that entries of $\mathfrak{B}_\mathbf{G}$ can be expressed in terms of $\kappa$-Galois inner products.
\begin{lemma}
\label{lem2}
Let $\mathbf{G}$ be a GPM of a $\Lambda$-MT code $\mathcal{C}$, where $\Lambda^{-p^{-\kappa}}=\Lambda$. For every $1\le i\le \ell$, we choose $\mathbf{a}_i\in\mathcal{C}$ so that 
\begin{equation*}
\mathrm{Row}_i\left(\mathbf{G}\right)=\phi\left(\mathbf{a}_i\right)+\mathbf{u}_i(x) \mathrm{diag}\left[x^{m_j}-\lambda_j\right]
\end{equation*}
for some polynomial vector $\mathbf{u}_i(x)$. For every pair $(u,v)$ with $1\le u,v \le \ell$,
\begin{equation*}
\mathrm{Ent}_{u,v}\left(\mathfrak{B}_\mathbf{G}\right)=\sum_{\iota=0}^{N-1}\langle \mathbf{a}_u ,\mathcal{T}^\iota_\Lambda\left(\mathbf{a}_v \right) \rangle_\kappa x^\iota.
\end{equation*}
\end{lemma}
\begin{proof}
The proof is essentially similar to the proof of Corollary \ref{corollary1}. It follows from Definition \ref{Def_of_B} that
\begin{equation*}
\begin{split}
\mathrm{Ent}_{u,v}\left(\mathfrak{B}_\mathbf{G}\right)&\equiv \mathrm{Row}_u\left(\mathbf{G}\right)\ \mathrm{diag}\left[ \frac{x^N-1}{x^{m_j}-\lambda_j}\right]  \sigma^\kappa\left(\mathrm{Row}_v \mathbf{G}\left(\frac{1}{x}\right) \mathrm{diag}\left[ x^{m_j}\right] \right)^t\\
&  \equiv \phi\left(\mathbf{a}_u \right) \mathrm{diag}\left[ \frac{x^N-1}{x^{m_j}-\lambda_j}\right]   \sigma^\kappa\left(\phi\left(\mathbf{a}_v\right)\left(\frac{1}{x}\right) \mathrm{diag}\left[ x^{m_j}\right] \right)^t \\
&  \equiv   \sum_{\iota=0}^{N-1}\langle\mathbf{a}_u,\mathcal{T}^\iota_\Lambda\left(\mathbf{a}_v\right) \rangle_\kappa x^\iota \pmod{\left(x^N-1\right)}.
\end{split}
\end{equation*}
Equality follows because each side is a polynomial of degree at most $N-1$.
\end{proof}
For a MT code with a GPM $\mathbf{G}$, we define a generator matrix $S_\mathbf{G}$ whose rows are those of $\mathbf{G}$ and all their shifts. This generator matrix will be crucial in our proof of the Galois hull dimension formula.
\begin{definition}
\label{def_matrix_S}
Let $\mathbf{G}$ be a GPM of a $\Lambda$-MT code $\mathcal{C}$ of length $n$, where $\Lambda^{-p^{-\kappa}}=\Lambda$. For every $1\le i\le \ell$, we choose $\mathbf{a}_i\in\mathcal{C}\subseteq \mathbb{F}_q^n$ so that 
\begin{equation*}
\mathrm{Row}_i\left(\mathbf{G}\right)=\phi\left(\mathbf{a}_i\right)+\mathbf{u}_i(x) \mathrm{diag}\left[x^{m_j}-\lambda_j\right]
\end{equation*}
for some polynomial vector $\mathbf{u}_i(x)$. Let $S_\mathbf{G}$ be the $N\ell \times n$ matrix over $\mathbb{F}_q$ defined by
\begin{equation*}
\mathrm{Row}_{\iota \ell +i}\left(S_\mathbf{G}\right)=\mathcal{T}_\Lambda^\iota \left(\mathbf{a}_i\right)
\end{equation*}
for $1\le i\le \ell$ and $0\le \iota\le N-1$. Furthermore, let $B_\mathbf{G}=S_\mathbf{G}\left( \sigma^\kappa S_\mathbf{G}\right)^t$, i.e.,
\begin{equation*}
\mathrm{Ent}_{u,v}\left(B_\mathbf{G}\right)=\langle \mathrm{Row}_u\left(S_\mathbf{G}\right),\mathrm{Row}_v\left(S_\mathbf{G}\right)\rangle_\kappa
\end{equation*}
for $1\le u,v \le N\ell$. We let $\mathcal{Q}_\mathbf{G}$ refer to the linear code over $\mathbb{F}_q$ of length $N\ell$ generated by $B_\mathbf{G}$. (In other words, $\mathcal{Q}_\mathbf{G}$ is the row span of $B_\mathbf{G}$.)
\end{definition}

\begin{lemma}
\label{lem4}
$\mathcal{Q}_\mathbf{G}$ is QC of length $N\ell$ and index $\ell$.
\end{lemma}
\begin{proof}
Let $\mathcal{F}$ be the linear transformation on $\mathbb{F}_q^{N\ell}$ that corresponds to $\ell$ cyclic shifts (cf. Equation \eqref{cyclic_shift_QC}). In the same convention as in Definition \ref{def_matrix_S}, it suffices to show that the row span of $B_\mathbf{G}$ is $\mathcal{F}$-invariant. First observe that
\begin{equation*}
\begin{split}
\mathrm{Ent}_{\iota_1 \ell +i_1, \iota_2 \ell +i_2}\left(B_\mathbf{G}\right)&=\langle \mathrm{Row}_{\iota_1 \ell +i_1}\left(S_\mathbf{G}\right),\mathrm{Row}_{\iota_2 \ell +i_2}\left(S_\mathbf{G}\right) \rangle_\kappa\\
&=\langle \mathcal{T}_\Lambda^{\iota_1} \left(\mathbf{a}_{i_1}\right), \mathcal{T}_\Lambda^{\iota_2} \left(\mathbf{a}_{i_2}\right) \rangle_\kappa\\
&=\langle \mathbf{a}_{i_1} , \mathcal{T}_\Lambda^{\iota_2-\iota_1} \left(\mathbf{a}_{i_2}\right) \rangle_\kappa \qquad\text{(By Lemma \ref{lem_shift})}\\
&=\mathrm{Ent}_{ i_1, (\iota_2-\iota_1) \ell +i_2}\left(B_\mathbf{G}\right).
\end{split}
\end{equation*}
Thus $\mathrm{Row}_{\iota \ell +i}\left(B_\mathbf{G}\right)=\mathcal{F}^\iota \left( \mathrm{Row}_{i}\left(B_\mathbf{G}\right)\right)$ and we are done. 
\end{proof}
So far a QC code $\mathcal{Q}_\mathbf{G}$ over $\mathbb{F}_q$ of length $N\ell$ and index $\ell$ has been associated to the given GPM $\mathbf{G}$ of a MT code $\mathcal{C}$. It is interesting to determine a GPM for $\mathcal{Q}_\mathbf{G}$ by means of some operations on $\mathbf{G}$. To this end, we begin by showing that the polynomial vector representations of rows of $B_\mathbf{G}$ are rows of $\mathfrak{B}_\mathbf{G}$. Analogous to our discussion of QC codes in Section \ref{Sec2}, we let $\varphi:\mathbb{F}_q^{N\ell}\rightarrow \left(\mathbb{F}_q[x]\right)^\ell$ be defined by
\begin{equation*}
\begin{split}
\varphi:\left(a_{0,1},a_{0,2},\ldots,a_{0,\ell},a_{1,1},a_{1,2},\ldots,a_{1,\ell},\ldots,a_{N-1,1},a_{N-1,2},\ldots,a_{N-1,\ell}\right)\\
\mapsto \left( \sum_{\iota=0}^{N-1}a_{\iota,1}x^\iota , \sum_{\iota=0}^{N-1}a_{\iota,2}x^\iota,\ldots, \sum_{\iota=0}^{N-1}a_{\iota,\ell}x^\iota\right).
\end{split}\end{equation*}

\begin{lemma}
\label{lem5}
$\varphi\left(\mathrm{Row}_i\left(B_\mathbf{G} \right)\right)=\mathrm{Row}_i\left(\mathfrak{B}_\mathbf{G} \right)$ for $1\le i\le \ell$.
\end{lemma}
\begin{proof}
In the language of Definition \ref{def_matrix_S}, we have
\begin{equation*}
\begin{split}
\varphi\left(\mathrm{Row}_{i}\left(B_\mathbf{G}\right)\right)&=\varphi\left(\left[\langle \mathrm{Row}_{i}\left(S_\mathbf{G} \right) , \mathrm{Row}_{\iota \ell +j}\left(S_\mathbf{G} \right) \rangle_\kappa\right]_{1\le j \le \ell \ \mathrm{and}\  0\le \iota \le N-1}\right)\\
&=\varphi\left(\left[\langle \mathbf{a}_i , \mathcal{T}_\Lambda^\iota \left(\mathbf{a}_j\right) \rangle_\kappa\right]_{1\le j \le \ell \ \mathrm{and}\  0\le \iota \le N-1}\right)\\
&=\left( \sum_{\iota=0}^{N-1}\langle \mathbf{a}_i , \mathcal{T}_\Lambda^\iota \left(\mathbf{a}_1\right) \rangle_\kappa x^\iota ,
 \sum_{\iota=0}^{N-1}\langle \mathbf{a}_i , \mathcal{T}_\Lambda^\iota \left(\mathbf{a}_2\right) \rangle_\kappa x^\iota,\ldots,
  \sum_{\iota=0}^{N-1}\langle \mathbf{a}_i , \mathcal{T}_\Lambda^\iota \left(\mathbf{a}_\ell\right) \rangle_\kappa x^\iota\right)\\
&=\mathrm{Row}_i\left(\mathfrak{B}_\mathbf{G} \right). \qquad \text{(By Lemma \ref{lem2})}
\end{split}
\end{equation*}
\end{proof}

The first main result is formalized in the following theorem, in which we use Theorem \ref{hongwei} to determine the dimension of the $\kappa$-Galois hull of a MT code from the dimension of the QC code $\mathcal{Q}_\mathbf{G}$. On the other hand, the dimension of $\mathcal{Q}_\mathbf{G}$ follows as a straightforward application of \eqref{Dim_QC} if a GPM for $\mathcal{Q}_\mathbf{G}$ is known. Thus, the following theorem also provides a mean of constructing such a GPM.
\begin{theorem}
\label{theorem1}
Let $q=p^e$ and let $0\le \kappa <e$. Assume that $\Lambda^{-p^{-\kappa}}=\Lambda$, where $\Lambda=\left(\lambda_1,\lambda_2,\ldots,\lambda_\ell\right)$. Let $\mathcal{C}$ be a $\Lambda$-MT code over $\mathbb{F}_q$ of block lengths $\left( m_1,m_2,\ldots,m_\ell\right)$. With the same notation as in Definitions \ref{Def_of_B} and \ref{def_matrix_S}, $\mathcal{Q}_\mathbf{G}$ is QC over $\mathbb{F}_q$ of length $N\ell$ and index $\ell$, and is generated as an $\mathbb{F}_q[x]$-module by the rows of the polynomial matrix
\begin{equation}
\label{long_GPM}
\begin{pmatrix}
\mathfrak{B}_\mathbf{G}\\
\left(x^N-1\right)\mathbf{I}_\ell
\end{pmatrix}.
\end{equation}
Furthermore, the dimension of $\mathcal{Q}_\mathbf{G}$ as an $\mathbb{F}_q$-vector space is $$\mathrm{dim}\left( \mathcal{Q}_\mathbf{G}\right)=\mathrm{dim}\left( \mathcal{C}\right)-\mathrm{dim}\left( h_\kappa\left(\mathcal{C}\right)\right).$$ 
\end{theorem}
\begin{proof}
We know from Lemma \ref{lem4} that $\mathcal{Q}_\mathbf{G}$ is certainly QC of length $N\ell$ and index $\ell$, while Lemma \ref{lem5} shows that
\begin{equation*}
\begin{split}
\varphi\left(\mathrm{Row}_{\iota \ell +i}\left(B_\mathbf{G}\right)\right)=\varphi\left(\mathcal{F}^\iota \left( \mathrm{Row}_{i}\left(B_\mathbf{G}\right)\right)\right)&\equiv x^\iota \varphi\left( \mathrm{Row}_{i}\left(B_\mathbf{G}\right)\right)\\ &= x^\iota \mathrm{Row}_i\left(\mathfrak{B}_\mathbf{G} \right) \pmod{\left(x^N-1\right)}
\end{split}
\end{equation*} 
for any $1\le i\le \ell$ and $0\le \iota\le N-1$. Therefore, $\mathcal{Q}_\mathbf{G}$ is the smallest QC code of index $\ell$ and co-index $N$ that contains the first $\ell$ rows of $B_\mathbf{G}$ (or equivalently, all rows of $\mathfrak{B}_\mathbf{G}$). The result on the polynomial matrix given by \eqref{long_GPM} follows immediately from Theorem \ref{GetGPM}. By Definition \ref{def_matrix_S}, $S_\mathbf{G}$ is a generator matrix of $\mathcal{C}$ as an $\mathbb{F}_q$-vector space. By an application of Theorem \ref{hongwei}, the rank of $B_\mathbf{G}$ (see Definition \ref{def_matrix_S}) is $\mathrm{dim}\left( \mathcal{C}\right)-\mathrm{dim}\left( h_\kappa\left(\mathcal{C}\right)\right)$. Then
\begin{equation*}
\mathrm{dim}\left(\mathcal{Q}_\mathbf{G}\right)=\mathrm{Rank}\left(B_\mathbf{G}\right)=\mathrm{dim}\left( \mathcal{C}\right)-\mathrm{dim}\left( h_\kappa\left(\mathcal{C}\right)\right).
\end{equation*}
\end{proof}

In fact, the class of $\kappa$-Galois self-orthogonal MT codes constitutes an interesting class of MT codes in which Theorem \ref{theorem1} has the following special form. It is noteworthy that a linear code is $\kappa$-Galois self-orthogonal if it is contained in its $\kappa$-Galois dual. On the contrary, a linear code is $\kappa$-Galois LCD if it trivially intersects its $\kappa$-Galois dual.
\begin{corollary}
\label{self-orth}
Given the same hypothesis and notation as Theorem \ref{theorem1}, $\mathcal{C}$ is $\kappa$-Galois self-orthogonal if and only if $\mathrm{dim}\left(\mathcal{Q}_\mathbf{G}\right)=0$ if and only if $\mathfrak{B}_\mathbf{G}= \mathbf{0}$.
\end{corollary}

\begin{example}
We continue with the GQC code $\mathcal{C}$ introduced in Example \ref{Ex_SO_GQC}. We use Corollary \ref{self-orth} to convince the reader that $\mathcal{C}$ is $3$-Galois self-orthogonal. Note that $N=12$, and from \eqref{Ex_GQC} we find that
\begin{equation*}\begin{split}
&\mathbf{G}\ \mathrm{diag}\left[ \frac{x^N-1}{x^{m_j}-\lambda_j}\right]  \sigma^\kappa\left(\mathbf{G}\left(\frac{1}{x}\right) \mathrm{diag}\left[ x^{m_j}\right] \right)^t \\
&=\left(x^{12}+1\right)\begin{pmatrix}
 x^2 + x &   x &  \theta^5 x\\
 x^2 &  x^2 + x  & \theta^5 x^2 + \theta^5 x\\
 \theta^{10} x^3 &  \theta^{10} x^3 + \theta^{10} x^2  & x^3 + x
\end{pmatrix}\equiv \mathbf{0}\pmod{\left(x^{12}+1\right)}. 
\end{split}\end{equation*}
Then $\mathfrak{B}_\mathbf{G}=\mathbf{0}$ by Definition \ref{Def_of_B}.
\end{example}

\begin{example}
Consider the binary QC code $\mathcal{C}$ of length $n=36$, index $\ell=6$, co-index $m=6$, and reduced GPM
\begin{equation*}
\mathbf{G}=\begin{pmatrix}
1 & 0 & 1 & x^{4} + x^{2} & f(x) & x^{5} + x^{3} + x + 1 \\
0 & 1 & 0 & x^{3} & 1 & x^{5} + x^{4} + x^{2} + x + 1 \\
0 & 0 & x + 1 & x^{4} + x^{3} + 1 & x^{4} + x^{3} & x^{5} + x^{2} + 1 \\
0 & 0 & 0 & f(x) & 0 & f(x) \\
0 & 0 & 0 & 0 & x^{6} + 1 & 0 \\
0 & 0 & 0 & 0 & 0 & x^{6} + 1
\end{pmatrix}
\end{equation*}
where $f(x)=x^{5} + x^{4} + x^{3} + x^{2} + x + 1$. From \eqref{Dim_QC}, $\mathrm{dim}\left(\mathcal{C}\right)=18$. Since $\mathcal{C}$ is QC, we see that $N=m=m_j=6$ and $\lambda_j=1$ for $1\le j\le \ell$. Thus, for $\kappa=0$, we have
\begin{equation*}
\mathbf{G}\ \mathrm{diag}\left[ \frac{x^N-1}{x^{m_j}-\lambda_j}\right]  \sigma^\kappa\left(\mathbf{G}\left(\frac{1}{x}\right)\mathrm{diag}\left[ x^{m_j}\right] \right)^t=
x^m \mathbf{G}\ \mathbf{G}^t\left(\frac{1}{x}\right). 
\end{equation*}
We observe that all entries of $x^m \mathbf{G}\ \mathbf{G}^t\left(\frac{1}{x}\right)$ are divisible by $x^6+1$. Hence by Definition \ref{Def_of_B}, $\mathfrak{B}_\mathbf{G}=\mathbf{0}$, and $\mathcal{C}$ is Euclidean ($\kappa=0$) self-orthogonal by Corollary \ref{self-orth}. On calculating the minimum Hamming distance $d_\mathrm{min}$ of $\mathcal{C}$, we noticed that $d_\mathrm{min}=8$. According to the database \cite{Grassl}, $\mathcal{C}$ is optimal as a linear code.
\end{example}

From Theorem \ref{theorem1} we know that the dimension of the Galois hull is determined by the dimension of the QC code $\mathcal{Q}_\mathbf{G}$. The latter dimension might require a GPM for $\mathcal{Q}_\mathbf{G}$, see \eqref{Dim_QC}. Such GPM can be obtained by reducing the polynomial matrix given in \eqref{long_GPM} to its Hermit normal form (cf. Corollary \ref{GetGPM_Corr1}). However, Corollary \ref{GetGPM_Corr2} suggests the following alternative method for determining the dimension $\mathrm{dim}\left(\mathcal{Q}_\mathbf{G}\right)$. Hence, the dimension $\mathrm{dim}\left( h_\kappa\left(\mathcal{C}\right)\right)$ can be computed from the determinantal divisors of $\mathfrak{B}_\mathbf{G}$.

\begin{theorem}
\label{det_div}
Let $q=p^e$ and let $0\le \kappa <e$. Assume that $\Lambda^{-p^{-\kappa}}=\Lambda$, where $\Lambda=\left(\lambda_1,\lambda_2,\ldots,\lambda_\ell\right)$. Let $\mathcal{C}$ be a $\Lambda$-MT code over $\mathbb{F}_q$ of block lengths $\left( m_1,m_2,\ldots,m_\ell\right)$. With the same notation as in Definition \ref{Def_of_B},  
\begin{equation*}\mathrm{dim}\left( h_\kappa\left(\mathcal{C}\right)\right)=\mathrm{dim}\left( \mathcal{C}\right)+\deg\left(\gcd_{0\le i\le \ell}\left\{\left(x^N-1\right)^{\ell-i}\mathfrak{d}_i\left(\mathfrak{B}_\mathbf{G}\right)\right\}\right)-N\ell\end{equation*}
where $\mathfrak{d}_i\left(\mathfrak{B}_\mathbf{G}\right)$ is the $i$-th determinantal divisor of $\mathfrak{B}_\mathbf{G}$.
\end{theorem}
\begin{proof}
Let $\mathfrak{G}_\mathbf{G}$ be a GPM for $\mathcal{Q}_\mathbf{G}$. For the same reason mentioned in the proof of Corollary \ref{GetGPM_Corr1}, there is an invertible $\mathbf{U}$ such that
\begin{equation*}
\mathbf{U}\begin{pmatrix}
\mathfrak{B}_\mathbf{G}\\
\left(x^N-1\right)\mathbf{I}_\ell
\end{pmatrix}=\begin{pmatrix}
\mathfrak{G}_\mathbf{G}\\
\mathbf{0}_{\ell\times\ell}
\end{pmatrix}.
\end{equation*}
Now we apply Corollary \ref{GetGPM_Corr2} and Theorem \ref{theorem1} to obtain
\begin{equation*}
\begin{split}
\mathrm{dim}\left( \mathcal{C}\right)-\mathrm{dim}\left( h_\kappa\left(\mathcal{C}\right)\right)=\mathrm{dim}\left( \mathcal{Q}_\mathbf{G}\right)=N \ell - \deg\left(\gcd_{0\le i\le \ell}\left\{\left(x^N-1\right)^{\ell-i}\mathfrak{d}_i\left(\mathfrak{B}_\mathbf{G}\right)\right\}\right).
\end{split}
\end{equation*}
\end{proof}
It readily follows from Theorems \ref{theorem1} and \ref{det_div} the following identification of LCD codes.
\begin{corollary}
\label{LCD}
Given the same notation as in Definitions \ref{Def_of_B} and \ref{def_matrix_S}, $\mathcal{C}$ is $\kappa$-Galois LCD if and only if $\mathrm{dim}\left(\mathcal{Q}_\mathbf{G}\right)=\mathrm{dim}\left(\mathcal{C}\right)$ if and only if 
$$\deg\left(\gcd_{0\le i\le \ell}\left\{\left(x^N-1\right)^{\ell-i}\mathfrak{d}_i\left(\mathfrak{B}_\mathbf{G}\right)\right\}\right)=N\ell-\mathrm{dim}\left( \mathcal{C}\right).$$
\end{corollary}

\begin{example}
Let $\ell=2$ and $m_1=m_2=2$. Consider the $\left(2,5\right)$-QT code $\mathcal{C}$ over $\mathbb{F}_7$ generated by $\{(x+1,0)\}$. In \cite[Example 5.3]{Sharma2018}, it has been demonstrated that $\mathcal{C}$ is Euclidean LCD. We aim to show this by applying Corollary \ref{LCD}. From the discussion that follows Definition \ref{def_MT}, a GPM for $\mathcal{C}$ can be constructed from the generating set $\{(x+1,0),(x^2-2,0),(0,x^2-5)\}$. As a consequence, the reduced GPM of $\mathcal{C}$ is
\begin{equation*}
\mathbf{G}=\begin{pmatrix}
1&0\\
0& x^2 + 2 
\end{pmatrix}.
\end{equation*}
Since $\kappa=0$ and $N=\mathrm{lcm}\left\{2\times3,2\times6\right\}=12$, it follows from Definition \ref{Def_of_B} that
\begin{equation*}\begin{split}
&\mathbf{G}\ \mathrm{diag}\left[ \frac{x^N-1}{x^{m_j}-\lambda_j}\right]  \sigma^\kappa\left(\mathbf{G}\left(\frac{1}{x}\right) \mathrm{diag}\left[ x^{m_j}\right] \right)^t =\begin{pmatrix}
x^{12} + 2 x^{10} + 4 x^8 + x^6 + 2 x^4 + 4 x^2 &  0\\
0 &  2 x^{14} + x^{12} + 5 x^2 + 6
\end{pmatrix}.
\end{split}\end{equation*}
Reduction modulo $(x^{12}-1)$ yields
\begin{equation*}
\mathfrak{B}_\mathbf{G}= \begin{pmatrix}
 2 x^{10} + 4 x^8 + x^6 + 2 x^4 + 4 x^2 +1 &  0\\
0 &   0
\end{pmatrix}. 
\end{equation*}
In this case, the Hermite normal form of the polynomial matrix given in \eqref{long_GPM} provides the reduced GPM of $\mathcal{Q}_\mathbf{G}$, which is given by
\begin{equation*}
\mathfrak{G}_\mathbf{G}=
\begin{pmatrix}
 2 x^{10} + 4 x^8 + x^6 + 2 x^4 + 4 x^2 +1 &  0\\
0 &   x^{12}+6
\end{pmatrix}.
\end{equation*}
Equation \eqref{Dim_QC} shows that $\mathrm{dim}\left(\mathcal{Q}_\mathbf{G}\right)=N\ell -\deg\left( \mathrm{det}\mathfrak{G}_\mathbf{G}\right)=24-22=2=\mathrm{dim}\left(\mathcal{C}\right)$. Thus $\mathcal{C}$ is LCD by Corollary \ref{LCD}. We could have shown this in an alternative way by using the determinantal divisors of $\mathfrak{B}_\mathbf{G}$. Specifically, $\mathfrak{d}_0\left(\mathfrak{B}_\mathbf{G}\right)=1$, $\mathfrak{d}_1\left(\mathfrak{B}_\mathbf{G}\right)=2 x^{10} + 4 x^8 + x^6 + 2 x^4 + 4 x^2 +1$ divides $x^{12}-1$, and $\mathfrak{d}_2\left(\mathfrak{B}_\mathbf{G}\right)=0$. Thus, $\mathcal{C}$ is Euclidean LCD because
\begin{equation*}
\deg\left(\gcd_{0\le i\le 2}\left\{\left(x^{12}-1\right)^{2-i}\mathfrak{d}_i\left(\mathfrak{B}_\mathbf{G}\right)\right\}\right)=22=N\ell-\mathrm{dim}\left( \mathcal{C}\right).
\end{equation*}
\end{example}

The following theorem constitutes the main result of this paper. It is significantly deeper than the result of Theorem \ref{det_div}, where we offer a GPM for the Galois hull of a MT code. Specifically, we prove that the product of a GPM for the $\kappa$-Galois dual of $\mathcal{Q}_\mathbf{G}$ and $\mathbf{G}$ yields a GPM for the $\kappa$-Galois hull of $\mathcal{C}$. 
\begin{theorem}
\label{main_theorem}
Assume that $\Lambda^{-p^{-\kappa}}=\Lambda$, where $\Lambda=\left(\lambda_1,\lambda_2,\ldots,\lambda_\ell\right)$, $0\ne \lambda_j\in\mathbb{F}_q$ for $1\le j\le \ell$, $q=p^e$, and $0\le \kappa <e$. For a $\Lambda$-MT code $\mathcal{C}$ over $\mathbb{F}_q$ of block lengths $\left( m_1,m_2,\ldots,m_\ell\right)$ and a GPM $\mathbf{G}$, we set 
\begin{equation*}
\mathfrak{B}_\mathbf{G}\equiv\mathbf{G}\ \mathrm{diag}\left[ \frac{x^N-1}{x^{m_j}-\lambda_j}\right] \sigma^\kappa\left(\mathbf{G}\left(\frac{1}{x}\right)\mathrm{diag}\left[ x^{m_j}\right] \right)^t \pmod{\left(x^N-1\right)}
\end{equation*}
such that  $\deg\left(\mathrm{Ent}_{i,j}\left(\mathfrak{B}_\mathbf{G}\right)\right)<N$ for $1\le i,j\le \ell$, where $N=\mathop{\mathrm{lcm}}_{1\le j\le \ell}\left(t_j m_j\right)$ and $t_j$ is the multiplicative order of $\lambda_j$. We denote the QC over $\mathbb{F}_q$ of length $N\ell$ and index $\ell$ generated by the rows of 
\begin{equation*}
\begin{pmatrix}
\mathfrak{B}_\mathbf{G}\\
\left(x^N-1\right)\mathbf{I}_\ell
\end{pmatrix}
\end{equation*}
by $\mathcal{Q}_\mathbf{G}$. If $\mathfrak{H}_\mathbf{G}$ is a GPM of the $\kappa$-Galois dual $\mathcal{Q}_\mathbf{G}^{\perp_\kappa}$ of $\mathcal{Q}_\mathbf{G}$, then $\mathfrak{H}_\mathbf{G}\mathbf{G}$ is a GPM for $h_\kappa\left(\mathcal{C}\right)$.
\end{theorem}
\begin{proof}
In the special case of the QC code $\mathcal{Q}_\mathbf{G}$ of index $\ell$ and co-index $N$, Corollary \ref{corollary1} asserts that
\begin{equation*}
\begin{split}
\mathbf{0} &\equiv  \mathfrak{H}_\mathbf{G}\ \sigma^{e-\kappa}\left( x^N \mathfrak{B}_\mathbf{G}^t\left(\frac{1}{x}\right)\right) \\
&\equiv  \mathfrak{H}_\mathbf{G} \mathbf{G} \mathrm{diag}\left[ x^{-m_j}\right] \sigma^{e-\kappa}\left( x^N \mathrm{diag}\left[ \frac{x^{-N}-1}{x^{-m_j}- \lambda_j}\right] \mathbf{G}^t\left(\frac{1}{x}\right)\right) \pmod{\left(x^N-1\right)}.
\end{split}\end{equation*}
Then there exists a polynomial matrix $\mathbf{N}$ such that
\begin{equation*}
\begin{split}
\mathfrak{H}_\mathbf{G} \mathbf{G} \mathrm{diag}\left[ x^{-m_j}\right] \sigma^{e-\kappa}\left( x^N \mathrm{diag}\left[ \frac{x^{-N}-1}{x^{-m_j}- \lambda_j}\right] \mathbf{G}^t\left(\frac{1}{x}\right)\right) =\left( x^N-1\right)\mathbf{N}.
\end{split}\end{equation*}
Therefore,
\begin{equation*}
\begin{split}
 \mathfrak{H}_\mathbf{G} \mathbf{G} \mathrm{diag}\left[ x^{-m_j}\right] \sigma^{e-\kappa}\left( x^N \mathrm{diag}\left[ \frac{x^{-N}-1}{x^{-m_j}- \lambda_j}\right] \mathbf{G}^t\left(\frac{1}{x}\right)\mathbf{A}^t\left(\frac{1}{x}\right)\mathrm{diag}\left[ x^{m_j}\right]\right)\\ =\left( x^N-1\right) \mathbf{N}\sigma^{e-\kappa}\left( \mathbf{A}^t\left(\frac{1}{x}\right)\mathrm{diag}\left[ x^{m_j}\right]\right)
\end{split}\end{equation*}
and
\begin{equation*}
\mathfrak{H}_\mathbf{G} \mathbf{G}  =- \mathbf{N}\sigma^{e-\kappa}\left( \mathbf{A}^t\left(\frac{1}{x}\right) \mathrm{diag}\left[ x^{m_j}\right] \right)= - \mathbf{N}\mathrm{diag}\left[ x^{d_j}\right] \mathbf{H}_\kappa.
\end{equation*}
This shows that rows of $\mathfrak{H}_\mathbf{G}\mathbf{G}$ are codewords of $\mathcal{C}^{\perp_\kappa}$, but since rows of $\mathfrak{H}_\mathbf{G}\mathbf{G}$ are codewords of $\mathcal{C}$, they are codewords of $h_\kappa\left(\mathcal{C}\right)$. Then for any GPM $\mathbf{L}_\kappa$ of $h_\kappa\left(\mathcal{C}\right)$, there is a polynomial matrix $\mathbf{U}$ such that $\mathfrak{H}_\mathbf{G}\mathbf{G}=\mathbf{U}\mathbf{L}_\kappa$.  We claim that $\mathfrak{H}_\mathbf{G}\mathbf{G}$ is a GPM for $h_\kappa\left(\mathcal{C}\right)$. This claim is proven once we show that $\mathbf{U}$ is invertible. But
\begin{equation*}
\begin{split}
\deg\left(\mathrm{det}\mathbf{U}\right)&=\deg\left(\mathrm{det}\mathfrak{H}_\mathbf{G}\right)+\deg\left(\mathrm{det}\mathbf{G}\right)-\deg\left(\mathrm{det}\mathbf{L}_\kappa\right)\\
&=\left(N\ell-\mathrm{dim}\mathcal{Q}_\mathbf{G}^{\perp_\kappa}\right)+\left(n-\mathrm{dim}\mathcal{C}\right)
-\left(n-\mathrm{dim}h_\kappa\left(\mathcal{C}\right)\right)\\
&=\mathrm{dim}\mathcal{Q}_\mathbf{G}-\mathrm{dim}\mathcal{C}+\mathrm{dim}h_\kappa\left(\mathcal{C}\right)\\
&=0. \qquad \text{ (By Theorem \ref{theorem1})}
\end{split}\end{equation*}
Thus $\mathbf{U}$ is invertible and the proof is complete.
\end{proof}

The following is a direct consequence of Theorem \ref{main_theorem} and is a supplementary result of Corollaries \ref{self-orth} and \ref{LCD}.
\begin{corollary}
\label{LCDand_self_orth}
Given the same hypothesis and notation as Theorem \ref{main_theorem}, $\mathcal{C}$ is $\kappa$-Galois self-orthogonal if and only if $\mathfrak{H}_\mathbf{G}\mathbf{G}$ is a GPM for $\mathcal{C}$ if and only if $\mathfrak{H}_\mathbf{G}$ is invertible. In contrast, $\mathcal{C}$ is $\kappa$-Galois LCD if and only if $\mathbf{A}$ is a GPM for $\mathcal{Q}_\mathbf{G}^{\perp_\kappa}$.
\end{corollary}

It is now worthwhile to present a detailed example of applying Theorem \ref{main_theorem}.
\begin{example}
\label{MDS}
Let $\Lambda=\left(1,\theta\right)$, where $\theta\in\mathbb{F}_4$ such that $\theta^2+\theta+1=0$. Consider the $\Lambda$-MT code $\mathcal{C}$ over $\mathbb{F}_4$ of index $2$, block lengths $\left(3,5\right)$, and reduced GPM
\begin{equation*}
\mathbf{G}=\begin{pmatrix}
1 & \theta^2 \\
0 & x + \theta^2
\end{pmatrix}.
\end{equation*}
The matrix that satisfies the identical equation \eqref{identicalMT} of $\mathbf{G}$ is
\begin{equation*}
\mathbf{A}=\begin{pmatrix}
x^{3} + 1 & \theta^2 x^{2} + \theta x + 1 \\
0 & x^{4} + \theta^2 x^{3} + \theta x^{2} + x + \theta^2
\end{pmatrix}.
\end{equation*}
We see from \eqref{dimensionMT} that $\mathrm{dim}\left(\mathcal{C}\right)=7$, we also found that the minimum Hamming distance of $\mathcal{C}$ is $d_\mathrm{min}=2$. Thus, $\mathcal{C}$ achieves the Singleton bound, and hence, $\mathcal{C}$ is MDS. Set $\kappa=1$. Our aim is to find a GPM for $h_\kappa\left(\mathcal{C}\right)$. Since $N=\mathrm{lcm}\left\{3, 15\right\}=15$, we find that
\begin{equation*}\begin{split}
&\mathbf{G}\ \mathrm{diag}\left[ \frac{x^N-1}{x^{m_j}-\lambda_j}\right]  \sigma^\kappa\left(\mathbf{G}\left(\frac{1}{x}\right)\mathrm{diag}\left[ x^{m_j}\right] \right)^t=\\
&\begin{pmatrix}
x^{12} + \theta x^{10} + x^{9} + x^{6} + \theta^2 x^{5} + x^{3} & x^{15} + \theta^2 x^{14} + \theta x^{10} + x^{9} + \theta^2 x^{5} + \theta x^{4} \\
\theta x^{16} + x^{15} + \theta^2 x^{11} + \theta x^{10} + x^{6} + \theta^2 x^{5} & \theta x^{16} + \theta^2 x^{14} + \theta^2 x^{11} + x^{9} + x^{6} + \theta x^{4}
\end{pmatrix}.
\end{split}
\end{equation*}
First, by Definition \ref{Def_of_B}, 
\begin{equation*}
\mathfrak{B}_\mathbf{G}=\begin{pmatrix}
x^{12} + \theta x^{10} + x^{9} + x^{6} + \theta^2 x^{5} + x^{3} & \theta^2 x^{14} + \theta x^{10} + x^{9} + \theta^2 x^{5} + \theta x^{4} + 1 \\
\theta^2 x^{11} + \theta x^{10} + x^{6} + \theta^2 x^{5} + \theta x + 1 & \theta^2 x^{14} + \theta^2 x^{11} + x^{9} + x^{6} + \theta x^{4} + \theta x
\end{pmatrix}.
\end{equation*}
Second, we find the reduced GPM $\mathfrak{G}_\mathbf{G}$ of $\mathcal{Q}_\mathbf{G}$ by reducing the polynomial matrix given by \eqref{long_GPM} to its Hermite normal form, namely, 
\begin{equation*}
\mathfrak{G}_\mathbf{G}=\begin{pmatrix}
x^{9} + \theta x^{7} + x^{6} + x^{3} + \theta^2 x^{2} + 1 & x^{12} + \theta^2 x^{11} + \theta x^{7} + x^{6} + \theta^2 x^{2} + \theta x \\
0 & x^{15} + 1 
\end{pmatrix}. 
\end{equation*}
The polynomial matrix that satisfies the identical equation $\mathfrak{A}_\mathbf{G} \mathfrak{G}_\mathbf{G}=(x^{15}+1)\mathbf{I}_2$ of $\mathfrak{G}_\mathbf{G}$ is (see \eqref{identicalQC})
\begin{equation*}
\mathfrak{A}_\mathbf{G}=\begin{pmatrix}
x^{6} + \theta x^{4} + x^{3} + \theta^2 x^{2} + 1 & x^{3} + \theta^2 x^{2} + \theta x \\
0 & 1
\end{pmatrix}.
\end{equation*}
Then $\mathrm{dim}\left(h_\kappa\left(\mathcal{C}\right) \right)=1$ by Theorem \ref{theorem1} because $\mathrm{dim}\left(\mathcal{Q}_\mathbf{G} \right)=6$ by \eqref{Dim_QC}. Third, we use $\mathfrak{A}_\mathbf{G}$ to get a GPM $\mathfrak{H}_\mathbf{G}$ for $\mathcal{Q}_\mathbf{G}^{\perp_\kappa}$ via the construction of Theorem \ref{GPMforDual}. In fact,
\begin{equation*}\begin{split}
\mathfrak{H}_\mathbf{G}&= \sigma^{e-\kappa}\left(x^{15}\ \mathfrak{A}_\mathbf{G} \left(\frac{1}{x}\right) \mathrm{diag}\left[x^{-9}, x^{-15}\right]\right)^t \pmod{x^{15}-1}\\
&=\begin{pmatrix}
x^{6} + \theta x^{4} + x^{3} + \theta^2 x^{2} + 1 & 0 \\
\theta^2 x^{14} + \theta x^{13} + x^{12} & 1
\end{pmatrix}.
\end{split}\end{equation*}
By Theorem \ref{main_theorem}, $\mathfrak{H}_\mathbf{G} \mathbf{G}$ is a GPM for $h_\kappa\left(\mathcal{C}\right)$. The reduced GPM of $h_\kappa\left(\mathcal{C}\right)$ is the Hermite normal form of $\mathfrak{H}_\mathbf{G} \mathbf{G}$ given as
\begin{equation*}
\begin{pmatrix}
x^{2} + \theta^2 x + \theta & \theta x^{4} + x^{3} + \theta^2 x^{2} + \theta x + 1 \\
0 & x^{5} + \theta
\end{pmatrix}.
\end{equation*}
\end{example}

We conclude this section with another example of different parameters.
\begin{example}
Consider the $\left(2\theta^2,1,2\right)$-MT code $\mathcal{C}$ over $\mathbb{F}_9$ of block lengths $\left(3,4,6\right)$ with the reduced GPM
\begin{equation*}
\mathbf{G}=\begin{pmatrix}
1 & 0 & \theta^6 x^{3} + \theta^2 x^{2} + \theta^3 x \\
0 & x + 2 & \theta^5 x^{3} + 2 x^{2} + \theta^6 x + \theta \\
0 & 0 & x^{4} + \theta^2 x^{3} + \theta^2 x + 2
\end{pmatrix}
\end{equation*}
where $\theta^2+2\theta+2=0$. According to \cite{Grassl}, $\mathcal{C}$ is near-optimal because $\mathrm{dim}\left(\mathcal{C}\right)=8$ and $d_{\mathrm{min}}=4$. We determine the $1$-Galois hull of $\mathcal{C}$ as follows. By calculations similar to those made in Example \ref{MDS}, the associated QC code $\mathcal{Q}_\mathbf{G}$ is of length $N\ell=36$, index $3$, and reduced GPM $\mathfrak{G}_\mathbf{G}=\left[\mathfrak{g}_{i,j}\right]$ where
\begin{center}
\begin{tabular}{||c|c||} 
 \hline
 $\left(i,j\right)$ & $\mathfrak{g}_{i,j}$ \\
 \hline\hline 
 $\left(1,1\right)$ & $x^{8} + x^{6} + 2 x^{2} + 2$ \\ 
 \hline
 $\left(1,2\right)$ & $\theta^6 x^{7} + \theta^7 x^{6} + 2 x^{5} + x^{4} + 2 x^{3} + x^{2} + \theta x + \theta^6$ \\
 \hline
 $\left(1,3\right)$ &  $x^{11} + x^{10} + \theta^7 x^{9} + \theta^2 x^{8} + \theta^2 x^{7} + \theta x^{6} + 2 x^{5} + 2 x^{4} + \theta^3 x^{3} + \theta^6 x^{2} + \theta^6 x + \theta^5$ \\
 \hline
 $\left(2,1\right)$ & $0$ \\
 \hline
 $\left(2,2\right)$ & $x^{9} + 2 x^{8} + x^{5} + 2 x^{4} + x + 2$ \\
 \hline
 $\left(2,3\right)$ & $0$ \\
 \hline
 $\left(3,1\right)$ & $0$ \\
 \hline
 $\left(3,2\right)$ & $0$ \\
 \hline
 $\left(3,3\right)$ & $x^{12} + 2$ \\  
 \hline
\end{tabular}
\end{center}
The polynomial matrix $\mathfrak{A}_\mathbf{G}$ such that $\mathfrak{A}_\mathbf{G} \mathfrak{G}_\mathbf{G}=(x^{12}-1)\mathbf{I}_3$ is given by
\begin{equation*}
\mathfrak{A}_\mathbf{G}=\begin{pmatrix}
x^{4} + 2 x^{2} + 1 & \theta^2 x^{2} + 2 x + \theta^6 & 2 x^{3} + 2 x^{2} + \theta^6 x + \theta^5 \\
0 & x^{3} + x^{2} + x + 1 & 0 \\
0 & 0 & 1
\end{pmatrix}.
\end{equation*}
Theorem \ref{GPMforDual} constructs the following GPM for the $1$-Galois dual of $\mathcal{Q}_\mathbf{G}$:
\begin{equation*}\begin{split}
\mathfrak{H}_\mathbf{G}&= \sigma^{e-\kappa}\left( x^{12}\ \mathfrak{A}_\mathbf{G} \left(\frac{1}{x}\right) \mathrm{diag}\left[x^{-8},x^{-9},x^{-12}\right]\right)^t \pmod{x^{12}-1}\\
&=\begin{pmatrix}
x^{4} + 2 x^{2} + 1 & 0 & 0 \\
\theta^2 x^{3} + 2 x^{2} + \theta^6 x & x^{3} + x^{2} + x + 1 & 0 \\
\theta^2 x^{11} + 2 x^{10} + 2 x^{9} + \theta^7 & 0 & 1
\end{pmatrix}.
\end{split}\end{equation*}
By Theorem \ref{main_theorem}, $\mathfrak{H}_\mathbf{G} \mathbf{G}$ is a GPM for $h_\kappa\left(\mathcal{C}\right)$ whose Hermite normal form is
\begin{equation*}
\begin{pmatrix}
x^{2} + \theta^2 x + 2 & 0 & x^{5} + \theta^2 x^{4} + 2 x^{3} + \theta^6 x^{2} + x + \theta^2 \\
0 & x^{4} + 2 & 0 \\
0 & 0 & x^{6} + 1
\end{pmatrix}.
\end{equation*}
\end{example}

\section{Conclusion}
\label{conclusion}
For a MT code $\mathcal{C}$ with a given GPM $\mathbf{G}$, some formulas have been proven to determine the dimension of the Galois hull of $\mathcal{C}$. Specifically, we associate a QC code $\mathcal{Q}_\mathbf{G}$ with $\mathbf{G}$, where it has been proven that the dimension of $\mathcal{Q}_\mathbf{G}$ is equal to the difference between the dimensions of $\mathcal{C}$ and its Galois hull. In addition, we proved that multiplying a GPM of the Galois dual of $\mathcal{Q}_\mathbf{G}$ and $\mathbf{G}$ produces a GPM for the Galois hull of $\mathcal{C}$. Two special cases were also investigated: self-orthogonal and complementary dual MT codes. 

As a motivation for future work, it might be useful to know whether these results are still valid for MT codes over chain rings. To this end, one can begin by examining the complexity of these results in the simplest subclass of MT codes, namely, the class of cyclic codes over chain rings.




\end{document}